\documentclass[letterpaper,11pt]{article}
\pdfoutput=1

\usepackage[utf8]{inputenc}

\usepackage{amsmath, amsthm, amssymb}
\usepackage[table,xcdraw]{xcolor}
\usepackage[vlined,linesnumbered]{algorithm2e}
\usepackage{mathtools}
\usepackage[numbers]{natbib}
\usepackage{comment} 
\usepackage{tcolorbox} 
\usepackage{xfrac}
\usepackage{hyperref}
\usepackage{multirow}
\usepackage{caption}
\usepackage{bm}
\usepackage{newfloat}
\usepackage{enumitem}
\usepackage{wrapfig}
\usepackage{hhline}
\usepackage{array, multirow}
\usepackage{amssymb}
\usepackage{pifont}

\usepackage[margin=1in]{geometry}

\allowdisplaybreaks

\definecolor{mygreen}{RGB}{20,150,100}
\definecolor{myred}{RGB}{150,0,0}

\hypersetup{
     colorlinks=true,
     citecolor= mygreen,
     linkcolor= myred,
     urlcolor= mygreen
}

\renewcommand{\epsilon}{\varepsilon}
\newcommand{\stackeq}[2]{\ensuremath{\stackrel{\text{#1}}{#2}}}

\DeclareMathOperator{\E}{\ensuremath{\normalfont \textbf{E}}}

\newcommand{\hiddencomment}[1]{}
\newcommand{\cmark}{\ensuremath{\checkmark}}%
\newcommand{\xmark}{\ensuremath{\times}}%

\newcommand{\MPC}[0]{{\normalfont MPC}}

\newcommand{\clusterspivot}[0]{\ensuremath{\mathcal{C}_{\text{PIV}}}}
\newcommand{\clustersalg}[0]{\ensuremath{\mathcal{C}_{r\text{-PIV}}}}
\newcommand{\pivotspivot}[0]{\ensuremath{P_{\text{PIV}}}}
\newcommand{\pivotsourpivot}[0]{\ensuremath{P_{r\text{-PIV}}}}

\newcommand{\pivot}[0]{\ensuremath{\textsc{Pivot}}}
\newcommand{\ourpivot}[0]{\ensuremath{r\textsc{-Pivot}}}
\newcommand{\costourpiv}[1]{\ensuremath{\textsc{COST}_{r\textsc{-PIV}}(#1)}}
\newcommand{\costpiv}[1]{\ensuremath{\textsc{COST}_{\textsc{PIV}}(#1)}}
\newcommand{\opt}[1]{\ensuremath{\textsc{OPT}(#1)}}

\DeclareMathOperator{\poly}{poly}

\newcommand{\mc}[1]{\ensuremath{\mathcal{#1}}}

\usepackage[noabbrev,nameinlink]{cleveref}
\crefname{lemma}{Lemma}{Lemmas}
\crefname{theorem}{Theorem}{Theorems}
\crefname{property}{Property}{Properties}
\crefname{claim}{Claim}{Claims}
\crefname{definition}{Definition}{Definitions}
\crefname{observation}{Observation}{Observations}
\crefname{assumption}{Assumption}{Assumptions}
\crefname{line}{Line}{Lines}
\creflabelformat{property}{(#1)#2#3}
\crefname{equation}{Eq}{Eq}
\creflabelformat{equation}{(#1)#2#3}
\crefname{section}{Section}{Sections}
\crefname{result}{Result}{Results}
\crefname{appendix}{Appendix}{Appendices}
\crefname{table}{Table}{Tables}
\crefname{myalgctr}{Algorithm}{Algorithms}

\newtheorem{theorem}{Theorem}[section]
\newtheorem{result}{Result}
\newtheorem{lemma}[theorem]{Lemma}

\newtheorem{corollary}[theorem]{Corollary}

\newtheorem{definition}[theorem]{Definition}
\newtheorem{claim}[theorem]{Claim}

\newtheorem{observation}[theorem]{Observation}

\definecolor{mylightgray}{RGB}{240,240,240}

\newenvironment{graytbox}{
\par\addvspace{0.1cm}
\begin{tcolorbox}[width=\textwidth,
                 boxsep=5pt,
                 left=1pt,
                 right=1pt,
                 top=2pt,
                 bottom=2pt,
                 boxrule=1pt,
                 arc=0pt,
                 colback=mylightgray,
                 colframe=black
                 ]
}{
\end{tcolorbox}
}

\newenvironment{whitetbox}{
\par\addvspace{0.1cm}
\begin{tcolorbox}[width=\textwidth,
                  boxsep=5pt,
                  left=1pt,
                  right=1pt,
                  top=2pt,
                  bottom=2pt,
                  boxrule=1pt,
                  arc=0pt,
                  colframe=black,
                  colback=white
                  ]
}{
\end{tcolorbox}
}

\newenvironment{myproof}{
\vspace{-0.5cm}
\begin{proof}
}{
\end{proof}
}

\newcounter{myalgctr}

\makeatletter
\renewcommand{\paragraph}{%
  \@startsection{paragraph}{4}%
  {\z@}{10pt}{-1em}%
  {\normalfont\normalsize\bfseries}%
}
\makeatother

\makeatletter
\patchcmd{\@algocf@start}
  {-1.5em}
  {0pt}
  {}{}
\makeatother

\title{
Almost 3-Approximate Correlation Clustering\\in Constant Rounds
}

\author{Soheil Behnezhad\thanks{Department of Computer Science, Stanford University.} \and Moses Charikar\footnotemark[1] \and Weiyun Ma\footnotemark[1] \and Li-Yang Tan\footnotemark[1]}

\date{}

\begin{document}

\maketitle

\begin{abstract}
    We study parallel algorithms for {\em correlation clustering}. Each pair among $n$ objects is labeled as either ``similar'' or ``dissimilar''. The goal is to partition the objects into arbitrarily many clusters while minimizing the number of disagreements with the labels. 
    
    \smallskip
    Our main result is an algorithm that for any $\epsilon > 0$ obtains a $(3+\epsilon)$-approximation in $O(1/\epsilon)$ rounds (of models such as massively parallel computation, local, and semi-streaming). This is a culminating point for the rich literature on parallel correlation clustering. On the one hand, the approximation (almost) matches a natural barrier of 3 for combinatorial algorithms. On the other hand, the algorithm's round-complexity is essentially constant.
    
    \smallskip
    To achieve this result, we introduce a simple $O(1/\epsilon)$-round parallel algorithm. Our main result is to provide an analysis of this algorithm, showing that it achieves a $(3+\epsilon)$-approximation. Our analysis draws on new connections to sublinear-time algorithms. Specifically, it builds on the work of \citet*{YoshidaYI-STOC09} on bounding the ``query complexity'' of greedy maximal independent set. To our knowledge, this is the first application of this method in analyzing the approximation ratio of any algorithm.
\end{abstract}

\thispagestyle{empty}
\clearpage
\setcounter{page}{1}

\clearpage

\section{Introduction}

We study parallel algorithms for the following {\em correlation clustering} problem. The input is a collection of objects and a complete labeling of the object-pairs as ``similar'' or ``dissimilar''. It would be convenient to model the labels by a graph $G=(V, E)$ where similar pairs are adjacent and dissimilar pairs are non-adjacent. The goal is to partition the vertex set $V$ into {\em arbitrarily} many clusters, capturing the labels as closely as possible. 
Unless $G$ is a collection of vertex-disjoint cliques, a perfect clustering satisfying all the labels does not exist. A natural objective is thus to find the clustering that minimizes {\em disagreements}\footnote{Another natural objective is to maximize agreements which is ``easier'' for approximate algorithms \cite{BansalBC-FOCS02,CharikarGW-FOCS03}.}: That is, the number of edges that go across clusters plus the number of non-adjacent pairs inside the clusters.

In contrast to some other clustering problems such as $k$-means, $k$-median, or $k$-center, correlation clustering does not require the number of clusters to be pre-specified. This, as well as the fact that correlation clustering uses information about both similarity and dissimilarity of the pairs in its output makes it a desirable clustering method for various tasks. Examples of applications of correlation clustering include image segmentation \cite{KimYNK14}, community detection \cite{ShiDELM21},  disambiguation tasks \cite{DBLP:journals/tkde/KalashnikovCMN08}, automated labeling \cite{AgrawalHKMT09,ChakrabartiKP08}, and document clustering \cite{BansalBC-FOCS02}, among others. 

In many of these applications, the input tends to be by orders of magnitude larger than what a single machine can handle. This has motivated a long and rich body of work studying efficient parallel algorithms for this problem; see  \cite{BlellochFS-SPAA12,ChierichettiDK-KDD14,PanPORRJ15,AhnCGMW-ICML15,FischerN-SODA18,CambusCMU-DISC21,Cohen-AddadLMNP21,AssadiChen-ITCS} and the references therein.

In this paper, we continue the study of parallel correlation clustering algorithms. Our main result is that an (almost) 3-approximation can be obtained in constant rounds.

\begin{graytbox}
\begin{result}[Informal -- see \Cref{thm:main}]\label{res:main}
    For any $\epsilon > 0$, one can obtain a $(3+\epsilon)$-approximation of correlation clustering in $O(1/\epsilon)$ parallel rounds.
\end{result}
\end{graytbox}

\cref{res:main} is a culminating point for parallel correlation clustering. First, its round complexity is essentially constant. Second, 3-approximation is a natural target; it remains the best achieved by any combinatorial correlation clustering algorithm, even sequential ones, that do not solve an LP.

To put our result into perspective, let us first overview the prior work. Throughout this paper, let $n = |V|$ denote the number of vertices in $G$, $m = |E|$ denote the number of edges in $G$, and $\Delta$ denote the maximum degree in $G$.

\subsection{State of Affairs on (Parallel) Correlation Clustering}

\paragraph{Sequential algorithms:} Correlation clustering was first introduced by \citet*{BansalBC-FOCS02,BansalBC-ML04} who showed that it admits a 
(large) constant approximation in polynomial time. Several follow-up works improved the approximation ratio \cite{CharikarGW-FOCS03,AilonCN-STOC05,AilonCN-JACM08,ChawlaMSY14}. The current best known is 2.06 by \citet*{ChawlaMSY14}, which is obtained by rounding the solution to a natural LP. It is also known that the problem is APX-hard \cite{CharikarGW-FOCS03}.

Among known combinatorial algorithms the best known approximation is 3. It is achieved by a surprisingly simple randomized algorithm of \citet*{AilonCN-STOC05}, known as the \pivot{} algorithm. It is worth noting that the \pivot{} algorithm is also useful in LP rounding; for instance \citet{ChawlaMSY14} first modify the input based on the LP solution, then run \pivot{} on the resulting graph. The \pivot{} algorithm iteratively picks a random vertex, clusters it with its remaining neighbors, removes this cluster from the graph, and recurses on the remaining graph. A slightly paraphrased, but still equivalent, variant of this algorithm reads as follows.



\begin{whitetbox}
\textbf{Algorithm} \pivot{} \cite{AilonCN-STOC05}\textbf{:}

\vspace{0.2cm}

\begin{itemize}[topsep=0pt, itemsep=0pt,leftmargin=13pt]
\item Draw a permutation $\pi$ of the vertex set $V$ uniformly at random.
\item While $G$ has at least one vertex:
\begin{itemize}[topsep=0pt,itemsep=0pt,leftmargin=10pt]
    \item Let $v$ be the vertex in $G$ with the lowest rank in $\pi$. Mark $v$ as a {\em pivot}.
    \item Put $v$ and its (remaining) neighbors in a cluster $C_v$ and remove the vertices of $C_v$ from $G$.
\end{itemize}
\end{itemize}
\end{whitetbox}

\paragraph{Parallel algorithms:} The algorithms noted above are all highly sequential. Even the simple \pivot{} algorithm, as stated, may take $\Omega(n)$ rounds as it picks only one pivot in each round. This has led to a rich and beautiful line of work on both parallelizing the \pivot{} algorithm and also devising new parallel algorithms from scratch \cite{BlellochFS-SPAA12,ChierichettiDK-KDD14,PanPORRJ15,AhnCGMW-ICML15,FischerN-SODA18,CambusCMU-DISC21,Cohen-AddadLMNP21,AssadiChen-ITCS}. See \cref{table:related} for an overview of these results in three models of {\em massively parallel computation} (MPC) with sublinear space, the {\em streaming} model with $\widetilde{O}(n)$ space, and the distributed local model. The formal definition of these models is not relevant for our discussion in this section, and is thus deferred to \cref{sec:implementations}.

\begin{table}[]
\centering
\begin{tabular}{|c|c|c|c|c|c|}
\hline
\rowcolor[HTML]{EFEFEF} 
\multicolumn{1}{|l|}{\cellcolor[HTML]{EFEFEF}Approx} & \multicolumn{1}{l|}{\cellcolor[HTML]{EFEFEF}Rounds} & \multicolumn{1}{l|}{\cellcolor[HTML]{EFEFEF}\parbox{1.5cm}{\centering \vspace{0.1cm} \small Sublinear\\ MPC \\[0.1cm]}} & \multicolumn{1}{l|}{\cellcolor[HTML]{EFEFEF}\parbox{1.5cm}{\centering \vspace{0.1cm} \small Semi-\\Streaming \\[0.1cm]}} & \multicolumn{1}{l|}{\cellcolor[HTML]{EFEFEF}Local} & \multicolumn{1}{l|}{\cellcolor[HTML]{EFEFEF}Reference} \\ \hline
$3$                                                  & $O(n)$                                              & \cmark                                                           & \cmark                                                 & \cmark                                             & \citet{AilonCN-STOC05,AilonCN-JACM08}                                 \\ \hline
$3+\epsilon$                                         & $O(\log n \cdot \log \Delta /\epsilon)$             & \cmark                                                           & \cmark                                                 & \cmark                                             & \citet{ChierichettiDK-KDD14}                           \\ \hline
$3$                                                  & $O(\log^2 n)$                                       & \cmark                                                           & \cmark                                                 & \cmark                                             & \citet{BlellochFS-SPAA12}                              \\ \hline
$3$                                                  & $O(\log n)$                                         & \cmark                                                           & \cmark                                                 & \cmark                                             & \citet{FischerN-SODA18,FischerN-TransAlg20}                                \\ \hline
$3$                                                  & $O(\log \Delta \cdot \log\log n)$                   & \cmark                                                           & \xmark                                                 & \xmark                                             & \citet{CambusCMU-DISC21}                               \\ \hline
$3$                                                  & $O(\log \log n)$                                    & \xmark                                                           & \cmark                                                 & \xmark                                             & \citet{AhnCGMW-ICML15,AhnCGMW-Alg21}                                 \\ \hline
$701$                                                & $O(1)$                                              & \cmark                                                           & \cmark                                                 & \cmark                                             & \citet{Cohen-AddadLMNP21}                              \\ \hline
$10^5$                                               & 1                                                   & \xmark                                                           & \cmark                                                 & \xmark                                             & \citet{AssadiChen-ITCS}                                \\ \hline\hline
$3+\epsilon$                                         & $O(1/\epsilon)$                                     & \cmark                                                           & \cmark                                                 & \cmark                                             & \textbf{This work.}                                    \\ \hline
\end{tabular}
\caption{\small Prior work on low-depth correlation clustering. Here $\epsilon > 0$ can be made arbitrarily small. See \cref{sec:implementations} for the formal definition of the models.}
\label{table:related}
\end{table}

\vspace{-0.2cm}
\paragraph{Parallelizing {\normalfont \pivot{}}:} An early work of \cite{ChierichettiDK-KDD14} was based on the idea of parallelizing \pivot{} through picking multiple pivots independently in each round. For the approximation analysis to go through, it is important for all vertices of each round to have nearly the same chance of getting marked as pivots. Additionally, since the pivots must be non-adjacent, not too many vertices can be marked as pivots in parallel, leading to a round-complexity of $O(\log^2 n / \epsilon)$ for a $(3+\epsilon)$-approximation.

Another beautiful line of work \cite{BlellochFS-SPAA12,FischerN-SODA18} on parallelizing \pivot{}, which is the closest to this paper in terms of techniques, is based on a parallel implementation of the so-called {\em randomized greedy maximal independent set} algorithm. This algorithm also picks multiple pivots in each round, but correlates these choices in a way that guarantees exactly the same output as \pivot{}. In this algorithm, one first fixes a random permutation $\pi$ as in \pivot{}. Then in each round all vertices that come earlier than their (remaining) neighbors in $\pi$ are marked as pivots in parallel, and are removed from the graph along with their neighbors. This repeats over the same permutation $\pi$ until the graph becomes empty.  It can be shown that the final set of pivots formed this way is exactly equal to the set of pivots found by \pivot{} over $\pi$. But because the decisions across different rounds are not independent, the parallel round complexity of this algorithm is more complicated to analyze.  \citet*{BlellochFS-SPAA12} were the first to show that this process w.h.p. terminates in $\poly\log n$ rounds. \citet*{FischerN-SODA18} later improved this to $O(\log n)$, which they proved is the correct bound for this algorithm by providing a matching lower bound of $\Omega(\log n)$.

Depending on the specific model of computation, the round-complexity of \pivot{} can be further improved to sublogarithmic \cite{AhnCGMW-ICML15,CambusCMU-DISC21}, e.g. to $O(\log \log n)$ in the semi-streaming model \cite{AhnCGMW-ICML15}. But all these algorithms still require $\omega(1)$ rounds. See \cref{table:related}.

\paragraph{Constant Round Parallel Algorithms:} An alternative line of work on parallel correlation clustering focuses on obtaining constant round algorithms \cite{Cohen-AddadLMNP21,AssadiChen-ITCS}. These algorithms do not attempt to parallelize \pivot{}. Rather, they are based on the key new insight that for an $O(1)$-approximation, it suffices to find clusters that are either singletons or near-cliques. This helps getting around the intricacies of finding a maximal independent set (as in \pivot{}) which in turn results in a much faster round complexity of $O(1)$. The main downside of solutions with only near-cliques and singleton clusters is that a large approximation is inherent to it (see \cite[Remark~3.10]{Cohen-AddadLMNP21} for an example). For instance, the approximations achieved by \cite{Cohen-AddadLMNP21}  and \cite{AssadiChen-ITCS} are respectively $701$ and $10^5$. 

Compared to the two lines of work discussed above on parallel correlation clustering, \cref{res:main} achieves the best of both worlds. Its approximation ratio comes close to the 3-approximation of the first line of work, and its round-complexity is essentially constant.

\vspace{-0.2cm}
\subsection{Our Contribution}\label{sec:ourcontribution}

We introduce a new algorithm \ourpivot{} which has a parameter $r \geq 1$ that adjusts the round-complexity. It proceeds in the same way as the parallel \pivot{} discussed above for $r$ rounds, then we truncate the process and no longer find pivots. As a result, unlike \pivot{}, not every vertex will have a pivot among its neighbors. Thus, we have to be careful about how we form the clusters at the end. For technical reasons, even the vertices that do have pivots among their neighbors may form singleton clusters in our algorithm. We will discuss the intuition behind this perhaps counter-intuitive process of forming the clusters later in \Cref{sec:techniques}.

\begin{whitetbox}
\textbf{Algorithm} \ourpivot{}\textbf{:} Our $r$-round algorithm for correlation clustering.

\vspace{0.2cm}

\begin{itemize}[topsep=0pt, itemsep=0pt,leftmargin=13pt]
\item Draw a permutation $\pi$ of the vertex set $V$ uniformly at random.
\item Initially every vertex is {\em unsettled}.
\item For $r$ rounds:
\begin{itemize}[topsep=0pt,itemsep=0pt,leftmargin=10pt]
    \item For any unsettled $v \in V$, mark $v$ as a {\em pivot} if $\pi(v) < \pi(u)$ for all unsettled $u \in N(v)$.
    \item Mark all pivots and any vertex adjacent to them as {\em settled}.
\end{itemize}
\item Every pivot starts a cluster which includes itself. Then for every non-pivot vertex $u$:
\begin{itemize}[topsep=0pt,itemsep=0pt,leftmargin=10pt]
    \item If there is no pivot in $N(u)$ or if there exists an unsettled vertex $w \in N(u)$ whose rank is smaller than all the pivots in $N(u)$, then $u$ forms a singleton cluster.
    \item Otherwise $u$ joins the cluster of the minimum rank pivot in $N(u)$.
\end{itemize}
\end{itemize}
\end{whitetbox}

\vspace{-0.2cm}
\paragraph{Notation:} Let us denote the cost paid by \ourpivot{} for parameter $r$, permutation $\pi$, and input graph $G$ by \costourpiv{G, \pi}. We write \costpiv{G, \pi} to denote the cost paid by \pivot{} run on permutation $\pi$. Let us also denote by $\opt{G}$ the optimal correlation clustering cost of graph $G$. 

Our main result is that our truncated algorithm \ourpivot{} achieves almost the same approximation as the full fledged \pivot{} algorithm, even if $r$ is a rather small constant:

\begin{theorem}[\textbf{Main Technical Result}]\label{thm:main}
 For any graph $G$ and any $r \geq 1$,
    $$
    \E_\pi[\costourpiv{G, \pi}] \leq \E_\pi[\costpiv{G, \pi}] + \frac{8}{2r-1} \cdot \opt{G}.
    $$
\end{theorem}

Our proof of \cref{thm:main} is fundamentally different from how the original \pivot{} algorithm was analyzed, and is based on a new connection to sublinear time algorithms. We elaborate more on the key intuitions behind our analysis in \cref{sec:techniques}.

Combined with the 3-approximation guarantee of \cite{AilonCN-JACM08} for algorithm \pivot{}, \cref{thm:main} implies:

\begin{corollary}\label{cor:3-apx}
For any graph $G$ and any $r \geq 1$,
    $$
    \E_\pi[\costourpiv{G, \pi}] \leq \left(3+\frac{8}{2r-1}\right) \cdot \opt{G}.
    $$
\end{corollary}

\paragraph{Implications:} \Cref{cor:3-apx} implies that running \ourpivot{} for $r = O(1/\epsilon)$ suffices for a $(3+\epsilon)$-approximation. This is useful because \ourpivot{} can be implemented in $O(r)$ rounds of various models. In particular, we obtain the following results; see \cref{sec:implementations} for models/implementations.

\begin{corollary}[MPC] \label{cor:MPC}
    For any $\epsilon > 0$, there is a randomized $O(1/\epsilon)$-round MPC algorithm that obtains a $(3+\epsilon)$-approximation of correlation clustering. The algorithm requires $O(n^\delta)$ space per machine where constant $\delta > 0$ can be made arbitrarily small and requires $O(m)$ total space.
\end{corollary}

\begin{corollary}[Streaming]\label{cor:streaming}
    For any $\epsilon > 0$, there is a randomized $O(1/\epsilon)$-pass streaming algorithm using $O(n \log n)$ bits of space that obtains a $(3+\epsilon)$-approximation of correlation clustering.
\end{corollary}

\begin{corollary}[Local]\label{cor:local}
    For any $\epsilon > 0$, there is a randomized $O(1/\epsilon)$-round local algorithm that obtains a $(3+\epsilon)$-approximation of correlation clustering using $O(\log n)$-bit messages.
\end{corollary}

\begin{corollary}[LCA]\label{cor:LCA}
    For any $\epsilon > 0$, there is a randomized local computation algorithm (LCA) that obtains a $(3+\epsilon)$-approximation of correlation clustering in $\Delta^{O(1/\epsilon)} \cdot \poly\log n$ time/space.
\end{corollary}

\paragraph{Instance-wise Guarantee:} It is worth emphasizing that our approximation guarantee of \ourpivot{} in \cref{thm:main} is {\em instance-wise} close to what \pivot{} achieves. That is, if for some input $G$ the \pivot{} algorithm obtains an $\alpha$-approximation where $\alpha < 3$, then the 3-factors of all corollaries above for \ourpivot{} also improve to $\alpha$ for this instance $G$. This is appealing for two main reasons:
\begin{itemize}[topsep=5pt,leftmargin=15pt]
    \item \textbf{Practical purposes:} If \pivot{} performs better than its worst-case guarantee on certain input distributions, then so does \ourpivot{}.
    \item \textbf{Rounding the natural LP in $O(1)$ rounds:} As discussed, \cite{ChawlaMSY14} obtain a  2.06-approximation by first (locally) modifying the graph based on the LP solution, and then running \pivot{} on the modified graph. Thus if we are given this LP solution in any of the models above, we can first modify the graph (this step is simple and local, so can be done efficiently in all these models) and then instead of  \pivot{} run \ourpivot{} on it. Because of the instance-wise guarantee, this also leads to an (almost) 2.06-approximation, but now in only $O(1)$ rounds.
\end{itemize}

\paragraph{Future Work:} Our work leaves several interesting questions especially in big data models where there is no direct notion of rounds. See \cref{sec:conclusion} for some of these open problems.

\clearpage

\section{A High-Level Overview of Our Techniques}\label{sec:techniques}

In this section, we give some high-level intuitions behind both our algorithm and its analysis.

It would be useful to first compare the outputs of \pivot{} and \ourpivot{} when both algorithms are run on the same permutation $\pi$. \Cref{fig:compare} provides such a comparision over an example for $r=1$. We write \clustersalg{} and \clusterspivot{} to denote the output clusters of \ourpivot{} and \pivot{} respectively, and use \pivotspivot{} and \pivotsourpivot{} to denote their pivots respectively.

\begin{figure}[h]
    \centering
    \includegraphics{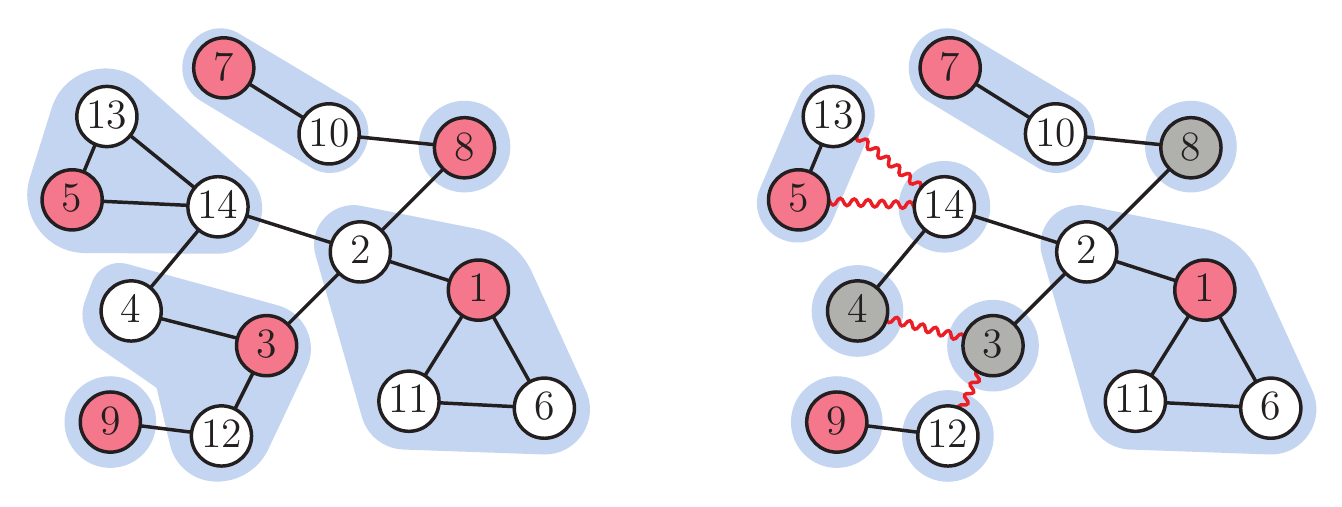}
    \caption{\small Comparison of \pivot{} (left) with \ourpivot{} (right) for $r=1$ over the same permutation $\pi$ whose ranks are shown on the vertices. The red vertices are pivots, the blue areas are the clusters, and the gray vertices are those marked as unsettled by \ourpivot{}. The red zigzagged edges are those that \ourpivot{} makes a mistake for but \pivot{} does not.}
    \label{fig:compare}
\end{figure}

\vspace{-0.2cm}
\paragraph{Some Intuition Behind {\normalfont \ourpivot{}}:} The final step of \ourpivot{}, where we form the clusters, is specifically designed to ensure that \clustersalg{} is a {\em refinement} of \clusterspivot{}. That is, each cluster in \clustersalg{} is completely inside a  cluster in \clusterspivot{} (see \Cref{fig:compare}). To achieve this, if some vertex $u$ joins the cluster of a pivot $v \in \pivotsourpivot{}$ in \ourpivot{}, we make sure that $u$ joins the cluster of $v$ in \pivot{} too. The unsettled vertices in \ourpivot{} are precisely defined to guarantee this property. Intuitively, while any settled vertex has the same pivot status in both \pivot{} and \ourpivot{}, unsettled vertices in \ourpivot{} may or may not be pivots in \pivot{}. Now if a vertex $u$ joins the cluster of a neighboring pivot $v$ in \ourpivot{}, we make sure that not only $v$ is the lowest rank pivot neighboring $u$, but that $v$'s rank is smaller than all unsettled neighbors of $u$ too. In \Cref{fig:compare}, e.g., even though vertex 12 has an adjacent pivot 9, it  decides not to join 9's cluster because of its unsettled neighbor 3. Note that 3 indeed ends up being a pivot in \pivot{}, so this decision was crucial for \clustersalg{} to be a refinement of \clusterspivot{}. On the other hand, vertex 14 does not join the cluster of 5 because of its unsettled neighbor 4, but this time 4 is not a pivot in \pivot{} which clusters 14 and 5 together. In our analysis, we have to make sure that our criteria which rather aggressively puts the vertices into singleton clusters does not hurt the approximation ratio of \ourpivot{} much, compared to \pivot{}.

\vspace{-0.1cm}
\paragraph{Analysis of {\normalfont \ourpivot{}}:} To analyze \ourpivot{}, we couple it with \pivot{} (over the same $\pi$) and  bound the number of vertex pairs that are mistakenly clustered by \ourpivot{} but correctly clustered by \pivot{}. Let $X$ be the set of such pairs. Our key contribution is to show that $\E|X| = O(\opt{G}/r)$ (stated as \Cref{lem:EX}). This is useful because if we run \ourpivot{} for only $r=O(1/\epsilon)$ steps, then we only pay an expected extra cost of $\E|X| \leq \epsilon \cdot \opt{G}$ compared to what \pivot{} pays, which is essentially the guarantee of our main \Cref{thm:main}. Below, we present the key ideas behind how we prove this upper bound on $\E|X|$.

Our first insight is that no pair in $X$ can be a non-edge. Note that the endpoints of any non-edge in $X$, by definition of $X$, should be clustered together in \clustersalg{} but separated in \clusterspivot{}. However, this would contradict our earlier discussion that \clustersalg{} is a refinement of \clusterspivot{}. Therefore, $X$ is essentially the set of edges in $E$ whose endpoints are separated by \clustersalg{}, but clustered together in \clusterspivot{}. In \Cref{fig:compare}, the set $X$ is illustarted by zigzagged red edges. 

To show how we relate $X$ to \opt{G}, let us first recall a standard framework of the literature in {\em charging bad triangles}. A bad triangle is a triplet of vertices $\{u, v, w\}$ such that two of the pairs $\{u, v\}, \{u, w\}, \{v, w\}$ are edges and one is a non-edge. Observe that no matter how the vertices of a bad triangle are clustered, at least one pair must be clustered incorrectly. Thus, if one finds $\beta$ edge disjoint bad triangles in $G$, then $\opt{G} \geq \beta$. This implies that to show $|X| = O(\opt{G}/r)$, it suffices to charge $\alpha = \Omega(r)$ bad triangles for every edge in $X$, and guarantee that these $|X| \cdot \alpha$ bad triangles are all edge disjoint. Instead of a deterministic charging scheme, it would be more convenient to pick the charged bad triangles randomly. Doing so and by generalizing the same argument, one can show that instead of full edge disjointness of bad triangles, it suffices to prove that each vertex pair belongs, in expectation, to at most one charged bad triangle.

\vspace{-0.1cm}
\paragraph{Key Idea I -- Charging Triangles Far Away:} Our charging scheme differs from those in the literature in that a mistake and the triangle that we charge it to may be far (at distance $\Omega(r)$) from each other. Previous charging schemes were all {\em highly local}, in that they charge any mistake to a bad triangle involving it. For instance, the 3-approximate analysis \cite{AilonCN-JACM08} of \pivot{} charges any mistake $\{u, v\}$ to the bad triangle $\{u, v, w\}$ where $w$ is the lowest rank pivot in $N(u) \cup N(v)$. It is impossible to analyze \ourpivot{} with a local charging schemes. The reason is that a local charging scheme can be shown to imply a deterministic upper bound of $O(n \cdot \opt{G})$ on the clustering cost which, for instance, holds for $\pivot{}$ and {\em every} $\pi$. However, this deterministic upper bound does not hold for \ourpivot{} which for some pathological permutation may have cost $\Omega(n^2/r)$ times $\opt{G}$. For more details about this, see \cref{sec:det-apx-ourpiv}.

\vspace{-0.1cm}
\paragraph{Key Idea II -- Connections to Sublinear Algorithms:} As discussed above, in our analysis we charge the mistakes $\{u, v\} \in X$ to triangles that are far away from $u, v$. To pick these triangles, we build on a query process developed originally for sublinear time algorithms \cite{NguyenO08,YoshidaYI-STOC09}. It is not hard to see that algorithm \pivot{} marks a vertex $v$ as a pivot iff there is no vertex $u$ adjacent to $v$ such that $\pi(u) < \pi(v)$ and $u$ is a pivot. Therefore, to determine whether a vertex $v$ is a pivot, it suffices to recursively query whether any of its lower rank neighbors are pivots. The total number of vertices that are (recursively) explored to answer such queries is known as the {\em query complexity} of random greedy maximal independent set (RGMIS), which is equivalent to the \pivot{} algorithm. In a beautiful result, \citet*{YoshidaYI-STOC09} showed that in any $n$-vertex $m$-edge graph, an average vertex $v$ has expected query complexity $O(m/n)$. This has been an influential work in the area of sublinear-time algorithms, where the goal is to explore a small part of the graph and estimate various global properties of it. In this work, we apply this method to a completely different context, and use it to analyze the approximation ratio of \ourpivot{}.

We first propose a natural analog of the vertex query process for pairs of vertices (\cref{sec:oracles}). Then to charge a mistake $\{u, v\} \in X$ to a bad triangle, we run this pair query process on $\{u, v\}$. We use $\{u, v\} \in X$ to show that there will be a moment when the stack of recursive calls by the query process, which will be a path in the graph, has size $\Theta(r)$. We then take the {\em first} moment that the stack gets to this size, and charge its last three vertices, which we prove must form a bad triangle. We then generalize the analysis of \cite{YoshidaYI-STOC09} in a non-trivial way (particularly by using several structural properties of the mistakes in $X$) to bound the number of times each vertex pair is charged, arriving at our final approximation guarantee of \cref{thm:main}. We give a more detailed high-level comparison between our analysis and \cite{YoshidaYI-STOC09} later in \cref{sec:YYI} after formalizing our charging scheme.

\vspace{-0.1cm}
\paragraph{Correlation Clustering vs MIS:} We finish this section with a brief comparison of correlation clustering with maximal independent set (MIS). Recall that the set \pivotspivot{} of the pivots formed by the \pivot{} algorithm is an MIS of $G$. When we truncate parallel \pivot{} after $r$ rounds, we still have an independent set \pivotsourpivot{} as our pivot-set but it is not necessarily maximal. In light of \cref{thm:main} which shows the clusters started by \pivotsourpivot{}, for some $r = O(1/\epsilon)$, are almost as good as the clusters started by \pivotspivot{}, it could be interesting to see how \pivotsourpivot{} and \pivotspivot{} compare in size.  It might be natural to go as far as conjecturing that $\E|\pivotsourpivot| \geq (1-\epsilon)\E|\pivotspivot|$ for $r = O(1/\epsilon)$. But this is far from the truth. In fact, there are graphs for which $\E|\pivotsourpivot|$ is $n^{\Theta(1/r)}/r$ times smaller than $\E|\pivotspivot|$; see \cref{sec:ourpiv-size}. Therefore, it is perhaps surprising that \pivotsourpivot{}, while being significantly smaller than \pivotspivot{}, is almost as good for starting clusters and approximating correlation clustering.

\section{The Analysis}

In this section we prove \Cref{thm:main} by analyzing the approximation ratio of \ourpivot{}.

We use \clusterspivot{} and \clustersalg{} to respectively denote the set of clusters returned by \pivot{} and \ourpivot{}. Let $X$ be the set of pairs of vertices $\{u, v\}$ such that \clustersalg{} disagrees with the label of $\{u, v\}$ but  \clusterspivot{} agrees with it. In other words, $X$ is the set of pairs for which \clustersalg{} pays an \underline{ex}tra cost for, compared to \clusterspivot{}. We have
\begin{equation}\label{eq:alg-cost}
    \costourpiv{G, \pi} \leq \costpiv{G, \pi} + |X|. 
\end{equation}

Our plan for proving \Cref{thm:main} is to show that the set $X$ has a small expected size. In particular, the core of our proof is the following bound on $X$ which immediately proves \Cref{thm:main}.

\begin{lemma}\label{lem:EX}
    For any $r \geq 1$, $\E_\pi|X| \leq \frac{8}{2r-1} \cdot \opt{G}$.
\end{lemma}

\begin{proof}[Proof of \Cref{thm:main} via \Cref{lem:EX}]
    By taking expectation over $\pi$ on both sides of \Cref{eq:alg-cost} and applying \Cref{lem:EX}, we get
    \begin{align*}
        \E_\pi[\costourpiv{G, \pi}] & \leq  \E_\pi[\costpiv{G, \pi}] + \E_\pi|X|\\
        & \leq \E_\pi[\costpiv{G, \pi}] + \frac{8}{2r-1} \cdot \opt{G}.\qedhere
    \end{align*}
\end{proof}

We prove \Cref{lem:EX} in the rest of this section.

\subsection{Charging Schemes}
We first recall a standard framework of the literature in {\em charging bad triangles} (see e.g. \cite{AilonCN-JACM08}). We say three distinct vertices $\{a, b, c\}$ in $V$ form a {\em bad triangle} if exactly two of the pairs $\{a, b\}, \{a, c\}, \{b, c\}$ belongs to $E$.

We say an algorithm $\mc{S}$ is a {\em charging scheme} for $X$, if $\mc{S}$ charges every pair in $X$ to a bad triangle of the input. We say $\mc{S}$ has {\em width} $w$ if for every pair of distinct vertices $a, b \in V$, the expected number of charges to bad triangles involving both $a$ and $b$ is upper bounded by $w$, where the expectation is taken over $\pi$ and the randomization of $\mc{S}$. We note that a triangle can be charged multiple times for different mistakes, but all of these charges must be counted in analyzing the width of the charging scheme. The following lemma shows that one can bound the expected size of $X$ in terms of the width of a charging scheme:

\begin{lemma}\label{lem:width-gives-apx}
If there exists a charging scheme $\mc{S}$ for $X$ that has width $w$, then
$$
    \E_\pi[|X|] \leq w \cdot \opt{G}.
$$
\end{lemma}

\Cref{lem:width-gives-apx} is a standard result in the framework of charging costs to bad triangles. For completeness, we provide a proof in \Cref{sec:width-proof}.

Therefore, to prove \Cref{lem:EX}, our plan is to design a charging scheme for $X$ that has width $\frac{8}{2r-1}$ for any $r \geq 1$. We define and analyze our charging scheme in \Cref{sec:our-CS}. To that end, we first give a characterization of the pairs in $X$ by comparing the clusterings \clusterspivot{} and \clustersalg{} in \Cref{sec:relate-clusters}. Then in \Cref{sec:oracles}, we associate bad triangles to the pairs in $X$ to be charged.

\subsection{Characterization of Pairs in $X$}\label{sec:relate-clusters}
In this subsection, we introduce some notations for the clusterings \clusterspivot{} and \clustersalg{} returned by \pivot{} and \ourpivot{} respectively. We show that \clustersalg{} is a certain refinement of \clusterspivot{} (\Cref{cl:cluster-refine}), and that all pairs in $X$ are edges in the graph (\Cref{cl:X-plus-pair}). This leads to our characterization of the pairs in $X$ (\Cref{lem:X-pair-classification}), which singles out an unsettled vertex with small rank in the neighborhood of each pair in $X$.

Throughout this subsection, we fix a permutation $\pi$ over $V$. 

\begin{definition}[Pivot sets]\label{def:pivotset}
Let \pivotspivot{} and \pivotsourpivot{} to respectively denote the set of pivots marked by \pivot{} and \ourpivot{} both run on the same permutation $\pi$.
\end{definition}

Clearly, $\pivotsourpivot{} \subseteq \pivotspivot{}$. Moreover, every cluster in \clusterspivot{} contains a unique pivot, and every cluster in \clustersalg{} that is not singleton contains a unique pivot.

\begin{definition}[Pivot of a vertex]\label{def:pivot}
For every $v \in V$, let $p_v \in \pivotspivot{}$ denote the pivot of the cluster in \clusterspivot{} that contains $v$. We say that $p_v$ is the pivot of $v$ in \pivot{}.
\end{definition}

Note that $p_v = v$ if and only if $v \in \pivotspivot{}$. Moreover, $\pi(p_v) \leq \pi(v)$ and $p_v$ is the vertex in $(N(v) \cup \{v\}) \cap \pivotspivot{}$ with minimum rank under $\pi$. Our next claim shows that the clustering \clustersalg{} refines \clusterspivot{} in a particular way. See \Cref{fig:compare}.

\begin{claim}\label{cl:cluster-refine}
For any cluster $C$ in \clusterspivot{}, if $v$ is the pivot of $C$, then $C$ is partitioned by distinct clusters $C_1', \dots, C_k'$ in \clustersalg{}, i.e.
\[
    C = C_1' \cup C_2' \cup \cdots C_k',
\]
such that $v \in C_1'$ and the remaining clusters $C_2', \dots, C_k'$ are all singletons. In particular, any cluster $C'$ in \clustersalg{} is contained in a cluster in \clusterspivot{}.
\end{claim}

\begin{myproof}
We first show that any cluster $C' \in \clustersalg{}$ that intersects $C$ but does not contain the pivot $v$ of $C$ must be singleton. Suppose otherwise that $C'$ is not singleton. Then $C'$ contains a unique pivot $v' \in \pivotsourpivot{}$, which must be different from $v$ since $v \not \in C'$. In particular, $v'$ cannot be contained in $C$ since otherwise $C$ would contain two different pivots $v$ and $v'$ in \pivotspivot{}. Now take a vertex $u \in C \cap C'$, which must be different from both $v$ and $v'$. Note that $v, v' \in N(u)$ and $v = p_u$. Thus we have $\pi(v) < \pi(v')$. Then, in \ourpivot{}, either $v$ is identified as a pivot, in which case $u$ would prefer to join the cluster of $v$ in the last step, or $v$ is unsettled, in which case $u$ would form a singleton cluster. Either contradicts that $u$ joins the cluster of $v'$ in \clustersalg{}.

It remains to show that the cluster $C_1' \in \clustersalg{}$ that contains $v$ is contained in $C$. Suppose otherwise that there exists a vertex $u \in C_1' \setminus C$. Then $u$ is contained in a cluster $C'' \in \clusterspivot{}$ other than $C$, and the pivot $v''$ of $C''$ is different from $v$. Since $C_1'$ already contains the pivot $v$, we have $v'' \not \in C_1'$. By the first part of the proof applied to $C''$, we see that $C_1'$ must be a singleton cluster that only contains $u$, which is a contradiction. Therefore, no such $u$ exists and $C_1' \subseteq C$.
\end{myproof}

As a consequence, we show that all pairs in $X$ are edges in the graph:

\begin{claim}\label{cl:X-plus-pair}
    $X \subseteq E$.
\end{claim}
\begin{myproof}
Suppose otherwise that there exists a pair $\{u,v\} \in X \setminus E$. Since $\clustersalg{}$ disagrees with the label of $\{u,v\}$, $u$ and $v$ must belong to the same cluster $C' \in \clustersalg{}$. But by \Cref{cl:cluster-refine}, $C'$ must be contained in a cluster $C \in \clusterspivot{}$, which implies that $\clusterspivot{}$ also disagrees with the minus label of $\{u,v\}$. This is a contradiction.
\end{myproof}

Therefore, for $\{u,v\} \in X$, since \clusterspivot{} agrees with the plus label of $\{u,v\}$, we have $p_u = p_v$.

\begin{definition}[Common pivot of a pair in $X$]\label{def:pivot-pair}
For $\{u,v\} \in X$, let $p_{\{u,v\}} = p_u = p_v \in \pivotspivot{}$ denote the common pivot of $u$ and $v$ in \pivot{}.
\end{definition}

Note that for every $\{u,v\} \in X$, $\pi(p_{\{u,v\}}) \leq \min\{\pi(u), \pi(v)\}$ and $p_{\{u,v\}}$ is the vertex in $(N(u) \cup N(v)) \cap \pivotspivot{}$ with minimum rank under $\pi$. Based on the last step of \ourpivot{}, we can characterize the pairs in $X$ as follows:

\begin{lemma}\label{lem:X-pair-classification}
For every $\{u,v\} \in X$, one of the following holds after $r$ rounds in \ourpivot{}:
\begin{enumerate}[label=\rm{(\roman*)}]

    \item $p_{\{u,v\}}$ is unsettled. \label{item:case-i}
    
    \item $p_{\{u,v\}}$ is settled. Moreover, at least one of $u$ and $v$ is singleton and has an unsettled, non-pivot neighbor $w_{\{u,v\}} \not \in \pivotspivot{}$ with $\pi(w_{\{u,v\}}) < \pi(p_{\{u,v\}})$. \label{item:case-ii}
    
\end{enumerate}
\end{lemma}

\begin{myproof}
It suffices to show that in Case (ii), at least one of $u$ and $v$ is singleton and has an unsettled neighbor $w_{\{u,v\}} \not \in \pivotspivot{}$ with $\pi(w_{\{u,v\}}) < \pi(p_{\{u,v\}})$. Let $C_u'$ (resp. $C_v'$) be the cluster in \clustersalg{} that contains $u$ (resp. $v$). Note that $C_u' \neq C_v'$.  By \Cref{cl:cluster-refine}, both $C_u'$ and $C_v'$ are contained in the cluster $C$ in \clusterspivot{} started by the pivot $p_{\{u,v\}}$. Then at least one of $C_u', C_v'$ cannot contain $p_{\{u,v\}}$, which must thus be singleton again by \Cref{cl:cluster-refine}. Say $C_v' = \{v\}$ is singleton. Since $p_{\{u,v\}}$ is settled, by the last step of \ourpivot{}, $v$ must have an unsettled neighbor $w_{\{u,v\}} \in N(v)$ with $\pi(w_{\{u,v\}}) < \pi(p_{\{u,v\}})$. In addition, $p_{\{u,v\}} = p_v$ is the pivot that $v$ joins in \ourpivot{}, which is the pivot neighbor of $v$ with minimum rank under $\pi$.  It follows that $w_{\{u,v\}} \not \in \pivotspivot{}$. 
\end{myproof}

\subsection{Vertex and Pair Oracles}\label{sec:oracles}
In this subsection, we define the following vertex and pair oracles. The vertex oracle locally determines whether a given vertex $v$ is part of the greedy MIS over permutation $\pi$, or equivalently in our context, whether $v$ is a pivot in the \pivot{} algorithm. The vertex oracle was first defined by \citet*{NguyenO08} and was further analyzed by \citet*{YoshidaYI-STOC09}. To analyze the pairs in $X$, we introduce a counterpart of the vertex oracle for pairs.

Throughout this subsection, we again fix a permutation $\pi$ over $V$.

\begin{whitetbox}
\begin{algorithm}[H]
    \DontPrintSemicolon
    \SetKwFunction{FnVx}{Vertex}
    \SetKwProg{Fn}{Function}{:}{}
    \Fn{\FnVx{$v$}}{
        Let $w_1, \ldots, w_d$ be vertices in $N(v)$ s.t. $\pi(w_1) < \ldots < \pi(w_d) < \pi(v)$.
	
        \For{$i$ in $1 \ldots d$}{
        	\lIf{\FnVx{$w_i$} $= 1$}{\Return 0}
        }
        \Return 1
    }
    
    \SetKwFunction{FnPr}{Pair}
    \SetKwProg{Fn}{Function}{:}{}
    \Fn{\FnPr{$u, v$}}{
        Let $w_1, \ldots, w_d$ be vertices in $N(u) \cup N(v)$ s.t. $\pi(w_1) < \ldots < \pi(w_d) < \min\{\pi(u), \pi(v) \}$.
	
        \For{$i$ in $1 \ldots d$}{
        	\lIf{\FnVx{$w_i$} $= 1$}{\Return 0}
        }
        \Return 1
    }
\end{algorithm}
\end{whitetbox}

\newcommand{\stack}[2]{\ensuremath{S_{#1}(#2)}}
\newcommand{\stackpath}[2]{\ensuremath{P_{#1}(#2)}}

For a vertex $v \in V$, \FnVx{$v$} returns $1$ if and only if $v$ is identified as a pivot by \pivot{}. As for a pair $\{u, v\}$ in $X$, it is straightforward to see that \FnPr{$u,v$} returns $1$ if and only if the common pivot $p_{\{u,v\}}$ of $u$ and $v$ in \pivot{} turns out to be one of $u, v$.

\begin{definition}\label{def:queryset}
We say that a vertex $v \in V$ (resp. pair $\{u,v\} \in X$) \emph{directly queries} a vertex $z \in V$ if \FnVx{$v$} (resp. \FnPr{$u,v$}) directly calls \FnVx{$z$}.
\end{definition}

We will be interested in the set of vertices directly queried by a vertex or a pair in $X$:

\begin{observation}\label{obs:vertex-query}
It holds that:
\begin{enumerate}[label=\rm{(\alph*)}]
    \item  \label{item:query-vertex-pivot}
    For a pivot $v \in \pivotspivot{}$, it directly queries all neighbors $z \in N(v)$ with rank $\pi(z) < \pi(v)$ and no other vertex. In particular, $v$ does not directly query any pivots in \pivotspivot{}.
    
    
    \item  \label{item:query-vertex-nonpivot}
    For a non-pivot $v \not \in \pivotspivot{}$, it directly queries all neighbors $z \in N(v)$ with rank $\pi(z) \leq \pi(p_v)$  and no other vertex. In particular, $p_v$ is the only pivot in \pivotspivot{} that $v$ directly queries.
    

    \item  \label{item:query-pair-pivot}
    For a pair $\{u,v\} \in X$ such that $p_{\{u,v\}} \in \{u,v\}$, it directly queries all neighbors $z \in N(u) \cup N(v)$ with rank $\pi(z) < \pi(p_{\{u,v\}})$ and no other vertex. In particular, $\{u,v\}$ does not directly query any pivots in \pivotspivot{}.
    

    \item  \label{item:query-pair-nonpivot}
    For a pair $\{u,v\} \in X$ such that $p_{\{u,v\}} \not \in \{u,v\}$, it directly queries all neighbors $z \in N(u) \cup N(v)$ with rank $\pi(z) \leq \pi(p_{\{u,v\}})$ and no other vertex. In particular, $p_{\{u,v\}}$ is the only pivot in \pivotspivot{} that $\{u,v\}$ directly queries.

    
\end{enumerate}
\end{observation}

In our analysis, we will focus on the stack of {\em recursive} calls to function \FnVx when we call \FnPr{$u, v$} for a pair $\{u, v\} \in X$. Our first insight is that, there is a moment when this stack includes $\Theta(r)$ elements.

\begin{claim}\label{cl:stack-gets-large}
    Let $\{u, v\} \in X$ and $\ell \leq 2r$. When we call \FnPr{$u,v$}, at some point the stack of recursive calls
    to \FnVx includes exactly $\ell$ elements.
\end{claim}

To prevent interruptions to the flow of this part, we defer the proof of \cref{cl:stack-gets-large} to \cref{sec:deferred}.

Of particular importance to our analysis, is the {\em first} moment that the stack of recursive calls to \FnVx when we call \FnPr{$u,v$} for $\{u,v\} \in X$ reaches a certain size.

\begin{definition}
Let $\{u,v\} \in X$ and consider the stack of recursive calls to \FnVx when we call \FnPr{$u,v$}. For each $\ell \in [2,2r]$, we denote by \stack{\ell}{u, v} the ordered list of the elements in the stack the first time that it includes $\ell$ elements.
\end{definition}

For example $\stack{4}{u, v} = (w_1, w_2, w_3, w_4)$ implies that $\{u, v\}$ directly queries $w_1$, $w_1$ directly queries $w_2$, $w_2$ directly queries $w_3$, $w_3$ directly queries $w_4$, and the first time after calling \FnPr{u, v} that the stack has 4 elements, only $w_1, w_2, w_3$, and $w_4$ are in it. \Cref{cl:stack-gets-large} ensures that \stack{\ell}{u, v} is defined for any $\{u,v\} \in X$ and $\ell \leq 2r$.

We now construct a path based on \stack{\ell}{u, v} by attaching $u$ and $v$ to the front in a particular order:

\begin{definition}
Let $\{u,v\} \in X$, $\ell \in [2,2r]$, and $\stack{\ell}{u, v} = (w_1, \dots, w_\ell)$. We define \stackpath{\ell}{u, v} to be the ordered list of $\ell+2$ vertices defined as:
\[
    \stackpath{\ell}{u,v} = \begin{cases}
        (v, u, w_1, \dots, w_\ell) & \text{if $w_1$ is directly queried by $u$ but not $v$,}\\
        (u, v, w_1, \dots, w_\ell) & \text{if $w_1$ is directly queried by $v$ but not $u$,}\\
        (v, u, w_1, \dots, w_\ell) & \text{if $w_1$ is directly queried by both $u$ and $v$, $\pi(u) < \pi(v)$.}
    \end{cases}
\]
\end{definition}

In other words, we choose the order of $u,v$ to ensure that $w_1$ is directly queried by the vertex that precedes it. If both $u$ and $v$ directly query $w_1$, we choose the order in a way that the ranks of the vertices in \stackpath{\ell}{u,v} are in descending order. Note that the first vertex $w_1$ in the stack \stack{\ell}{u, v} is directly queried by $\{u,v\}$, and it follows from \Cref{obs:vertex-query} that $w_1$ is directly queried by either $u$ or $v$. Thus \stackpath{\ell}{u,v} is always defined. Moreover, any two consecutive vertices in \stackpath{\ell}{u,v} form an edge in $E$, so it indeed specifies a path in $G$. The following observation is immediate from the construction of \stack{\ell}{u, v} and \stackpath{\ell}{u, v}:

\begin{observation}\label{obs:query-path}
Let $\{u,v\} \in X$, $\ell \in [2,2r]$, and write
\[
    \stackpath{\ell}{u,v} = (w_1, \dots, w_{\ell+2}).
\]
Then for each $i \in [3,\ell+2]$, we have $\pi(w_i) < \pi(w_{i-1})$, and $w_{i-1}$ directly queries $w_i$.
\end{observation}

Our charging scheme for $X$ relies crucially on the following claim:

\begin{claim}\label{cl:last-3-BT}
    For $\{u,v\} \in X$ and $\ell \in [2,2r]$, the last three vertices in \stackpath{\ell}{u, v} form a bad triangle.
\end{claim}

\begin{myproof}
We write
\[
    \stackpath{\ell}{u, v} = (w_1, w_2, \dots, w_{\ell+2}).
\]
To show that $\{w_\ell, w_{\ell+1}, w_{\ell+2}\}$ form a bad triangle, it suffices to show that $(w_\ell, w_{\ell+2}) \not \in E$. We assume otherwise that $(w_\ell, w_{\ell+2}) \in E$ and argue by contradiction. Note by \Cref{obs:query-path} that 
\[
    \pi(w_{\ell}) > \pi(w_{\ell+1}) > \pi(w_{\ell+2}),
\]
$w_\ell$ directly queries $w_{\ell+1}$, and $w_{\ell+1}$ directly queries $w_{\ell+2}$.

Suppose first that $w_{\ell+2} \in \pivotspivot{}$ is a pivot. Since $w_{\ell+2}$ is also a neighbor of $w_\ell$, its rank under $\pi$ must be at least that of $p_{w_\ell}$. But then $\pi(w_{\ell+1}) > \pi(w_{\ell+2}) \geq \pi(p_{w_\ell})$, which by \Cref{obs:vertex-query} \ref{item:query-vertex-nonpivot} implies $w_\ell$ does not directly query $w_{\ell+1}$. This is a contradiction.

Suppose on the other hand that $w_{\ell+2} \not \in \pivotspivot{}$. Then by \Cref{obs:vertex-query} \ref{item:query-vertex-nonpivot}, $w_{\ell+2}$ directly queries $p_{w_{\ell+2}}$. Moreover, since $w_\ell$ directly queries $w_{\ell+1}$ and $w_{\ell+2}$ is a neighbor of $w_\ell$ with rank smaller than $w_{\ell+1}$, $w_\ell$ also directly queries $w_{\ell+2}$. Therefore, when we call \FnPr{$u,v$}, at some point the stack of recursive calls to \FnVx consists of the following $\ell$ elements:
\[
    (w_3, \dots, w_{\ell}, w_{\ell+2}, p_{w_{\ell+2}}),
\]
which happens before the stack consists of $\stack{\ell}{u,v} = (w_3, \dots,w_{\ell}, w_{\ell+1}, w_{\ell+2})$. This contradicts that \stack{\ell}{u,v} is the list of elements in the stack when it {\em first} reaches $\ell$ elements.
\end{myproof}

\subsection{Our Charging Scheme}\label{sec:our-CS}
We can now define our charging scheme for $X$.

\begin{whitetbox}
\textbf{The Charging Scheme}\textbf{:}

\vspace{0.2cm}

\begin{itemize}[topsep=0pt, itemsep=0pt,leftmargin=13pt]
    \item Pick $\ell$ from $[2, 2r]$ uniformly at random.
    
    \item For any $(u, v) \in X$ charge the bad triangle formed by the last three vertices of \stackpath{\ell}{u,v}.
    
    (Here the existence of \stackpath{\ell}{u, v} (or \stack{\ell}{u, v}) follows from \Cref{cl:stack-gets-large} and its last three vertices form a bad triangle by \Cref{cl:last-3-BT}.)
\end{itemize}
\end{whitetbox}

Our main result is the following bound on the width of our charging scheme:

\begin{lemma}\label{lem:width}
Our charging scheme for $X$ has width $\frac{8}{2r-1}$ for any $r \geq 1$.

\end{lemma}

Note that by \Cref{lem:width-gives-apx}, \Cref{lem:width} immediately proves \Cref{lem:EX}.

Next, we focus on proving \Cref{lem:width}. 

\subsection{The High Level Approach and Relation to \cite{YoshidaYI-STOC09}}\label{sec:YYI}

We first give a high-level summary of our plan. To bound the number of charges to bad triangles involving two fixed vertices $a,b \in V$, we are led to bound the number of pairs of a permutation $\pi$ of $V$ and a path \stackpath{\ell}{u,v}, initiated by the pair oracle \FnPr called on a mistake $\{u,v\}$ in $X$ under $\pi$, that ends at a bad triangle involving both $a$ and $b$. To do this, we construct a new permutation $\tilde{\pi}$ from $\pi$ by rotating the ranks of vertices on the path \stackpath{\ell}{u,v} in a certain way while fixing the ranks of the other vertices. Let us temporarily denote $\tilde{\pi} = \phi(\pi, \stackpath{\ell}{u,v})$.

In the case $\{a,b\}$ is an edge in $E$, the map $\phi$ (illustrated as \Cref{fig:plus-rotation}) is the same as a construction by \citet*{YoshidaYI-STOC09}, which they used to bound the number of times the recursive query process of the vertex oracle \FnVx passes through a fixed edge $e \in E$, and thereby bound the query complexity of the vertex oracle. The key property of this construction is that the preimage of any permutation $\tilde{\pi}$ under the map $\phi$ can only have a small constant size. For \cite{YoshidaYI-STOC09}, this implies that in expectation over $\pi$ there are only $O(1)$ recursive queries that pass through $e$. For us, this implies that in expectation over $\pi$ there are only $O(1)$ possibilities for \stackpath{\ell}{u,v}. However, in our case, we need to consider the pair oracle in order to charge the pairs in $X$, and a path \stackpath{\ell}{u,v} does not necessarily specify a valid sequence of recursive queries initiated by the vertex oracle. There in fact may be many more than $O(1)$ such paths ending at a single edge. To get around this, we rely crucially on two additional properties in our case. First, each path \stackpath{\ell}{u,v} starts at a extra mistake $\{u,v\} \in X$ made by \ourpivot{} as compared to \pivot{}, which in particular is an {\em edge} whose two endpoints share a {\em common pivot} in \pivot. This provides new incidence relations among the vertices that help us rule out certain possibilities for \stackpath{\ell}{u,v}. Second, each path \stackpath{\ell}{u,v} is given by the {\em first} moment when the corresponding stack of recursive calls to \FnVx reaches a certain size. This places additional constraints on the preimages of $\phi$.

An additional challenge in our case as compared to \cite{YoshidaYI-STOC09} is that, we also need to consider pairs $\{a,b\}$ that are non-edges. In that case, we propose a modified construction of the rotation map $\phi$ (illustrated as \Cref{fig:minus-rotation}), which we show preserves the key property that the preimage of any permutation has a small constant size. Our argument for this property builds on that for the case where $\{a,b\} \in E$ but also carefully rules out several new edge cases, again by using the two additional properties in our case.

\subsection{Proof of \cref{lem:width}}

Now we start the proof. Given distinct vertices $a, b \in V$ and a permutation $\pi$ over $V$, we introduce a set $R_\pi(a,b)$ that accounts for the charges under $\pi$ to bad triangles involving both $a$ and $b$ that come from paths passing through $a$ before $b$. Formally, we define $R_{\pi}(a,b)$ to be the set of $(\{u,v\}, \ell)$, where $\{u,v\} \in X$ and $\ell \in [2, 2r]$ is an integer, such that under $\pi$:
\begin{itemize}
    \item the last three vertices in the ordered list \stackpath{\ell}{u,v} are $(a,b,c)$ or $(c,a,b)$ for some $c \in V$, if $\{a,b\} \in E$;
    
    \item the last three vertices in the ordered list \stackpath{\ell}{u,v} are $(a,c,b)$ for some $c \in V$, if $\{a,b\} \not \in E$.
\end{itemize}
To bound the number of charges to bad triangles involving both $a$ and $b$ and prove \Cref{lem:width}, we bound the size of $R_{\pi}(a,b)$ in the following two lemmas:

\begin{lemma}\label{lem:querypathplus}
Let $a, b \in V$ be distinct vertices such that $\{a,b\} \in E$. Then, for any $r \geq 1$, we have
\[
    \E_\pi[|R_\pi(a,b)|] \leq 4.
\]
\end{lemma}

\begin{lemma}\label{lem:querypathminus}
Let $a, b \in V$ be distinct vertices such that $\{a,b\} \not \in E$. Then, for any $r \geq 1$, we have
\[
    \E_\pi[|R_\pi(a,b)|] \leq 2.
\]
\end{lemma}

We first show that \Cref{lem:querypathplus} and \Cref{lem:querypathminus} together imply \Cref{lem:width}:

\begin{proof}[Proof of \Cref{lem:width} via \Cref{lem:querypathplus} and \Cref{lem:querypathminus}]
Let $a, b \in V$ be distinct vertices. We show that the expected number of charges to bad triangles involving both $a$ and $b$ is upper bounded by $\frac{8}{2r-1}$, where the expectation is taken over $\pi$ and $\ell$. In our charging scheme, a charge to a bad triangle involving both $a$ and $b$ comes from an ordered list \stackpath{\ell}{u,v} under some permutation $\pi$ and for some $\{u,v\} \in X$ and $\ell \in [2, 2r]$, such that the bad triangle consists of the last three vertices in \stackpath{\ell}{u,v}. In the case $\{a,b\} \in E$, $a$ and $b$ are among the last three vertices in \stackpath{\ell}{u,v} and are adjacent, although they can appear in either order. Since $\ell$ is chosen uniformly at random, by \Cref{lem:querypathplus}, the expected number of charges to bad triangles involving both $a$ and $b$ is at most
    \[
        \frac{1}{2r-1}\left(\E_\pi[|R_\pi(a,b)|] + \E_\pi[|R_\pi(b,a)|] \right) \leq \frac{8}{2r-1}.
    \]
In the other case $\{a,b\} \not \in E$, $a$ and $b$, in either order, are the last and third-to-last vertices in \stackpath{\ell}{u,v} respectively. Since $\ell$ is chosen uniformly at random, by \Cref{lem:querypathminus}, the expected number of charges to bad triangles involving both $a$ and $b$ is at most
    \[
        \frac{1}{2r-1}\left(\E_\pi[|R_\pi(a,b)|] + \E_\pi[|R_\pi(b,a)|] \right) \leq \frac{4}{2r-1} < \frac{8}{2r-1}.\qedhere
    \]
\end{proof}

\subsubsection{Proof of \Cref{lem:querypathplus}}\label{sec:proofplus}
Throughout this subsection, we fix distinct vertices $a, b \in V$ such that $\{a,b\} \in E$. Given a permutation $\pi$ over $V$ and $i \in [n]$, we denote by $\pi_i \in V$ the vertex with rank $i$ under $\pi$.

In our analysis, we break down the charges to bad triangles involving both $a$ and $b$ by specifying the {\em ranks} of the pair of vertices in $X$ that initiates the charge. Given $i,j \in [n]$ with $i<j$, we introduce a set $T^{i,j}(a,b)$ that accounts for the charges that come from paths originating at vertices with ranks $i$ and $j$ under some permutation $\pi$ and passing through $a$ before $b$. Formally, we define $T^{i,j}(a,b)$ to be the set of pairs $(\pi, \ell)$ of a permutation $\pi$ and an integer $\ell \in [2, 2r]$ such that under $\pi$, $\{\pi_i, \pi_j\} \in X$ and the last three vertices in the ordered list \stackpath{\ell}{\pi_i, \pi_j} are $(a,b,c)$ or $(c,a,b)$ for some $c \in V$. We will prove the following bound on the size of $T^{i,j}(a,b)$:

\begin{claim}\label{cl:boundpairplus}
For any $i,j \in [n]$ with $i<j$ and any $r \geq 1$, we have
\[
    |T^{i,j}(a,b)| \leq 8(n-2)!.
\]
\end{claim}

We first use \Cref{cl:boundpairplus} to prove \Cref{lem:querypathplus}:

\begin{proof}[Proof of \Cref{lem:querypathplus} via \Cref{cl:boundpairplus}]
Our main observation is that, as we range through all possibilities for the ranks $i$ and $j$ of the pair that initiates the charge, the sets $T^{i,j}(a,b)$ form a partition of the union of the sets $R_{\pi}(a,b)$ as $\pi$ ranges through all permutations. To make this precise, we define
\[
    R(a,b) = \{(\pi, \{u,v\}, \ell) \mid (\{u,v\}, \ell) \in R_\pi(a,b) \},
\]
\[
    T(a,b) = \{(i,j, \pi, \ell) \mid i < j, (\pi, \ell) \in T^{i,j}(a,b)\}.
\]
Then, it is straightforward to see that the map
\[
    R(a,b) \to T(a,b), \quad (\pi, \{u,v\}, \ell) \mapsto (\min\{\pi(u), \pi(v)\}, \max\{\pi(u), \pi(v)\}, \pi, \ell).
\]
is a bijection. By \Cref{cl:boundpairplus}, for any $r \geq 1$,
\[
    \sum_{\pi} |R_\pi(a,b)|  = |R(a,b)| = |T(a,b)| = \sum_{i<j} |T^{i,j}(a,b)|  \leq {n \choose 2} 8(n-2)! = 4n!,
\]
which implies that
\[
    \E_\pi[|R_\pi(a,b)|] \leq \frac{4n!}{n!} = 4.\qedhere
\]
\end{proof}

\begin{proof}[Proof of \Cref{cl:boundpairplus}]
We define a map
\[
    \phi^{i,j}_{a,b}: T^{i,j}(a,b) \to U^{i,j}_{a,b},
\]
where $U^{i,j}_{a,b}$ is a set of permutations $\tilde{\pi}$ over $V$ defined by
\[
    U^{i,j}_{a,b} = \{ \tilde{\pi} \mid \{\tilde{\pi}_i, \tilde{\pi}_j\} = \{a, b\}\}.
\]
Note that $U^{i,j}_{a,b}$ has size $2(n-2)!$. For $(\pi, \ell) \in T^{i,j}(a,b)$, the permutation $\phi^{i,j}_{a,b}(\pi, \ell)$ is defined by rotating the ranks of the vertices along the prefix $P$ of \stackpath{\ell}{\pi_i, \pi_j} ending at the directed edge $(a,b)$. Formally, we write
\[
    P = (w_1, w_2, \dots, w_{\bar{\ell}-2}, w_{\bar{\ell}-1} = a, w_{\bar{\ell}} = b),
\]
where $\{w_1, w_2\} = \{\pi_i, \pi_j\}$. If $(a,b)$ is the last edge on \stackpath{\ell}{\pi_i, \pi_j}, then $P$ is the same as \stackpath{\ell}{\pi_i, \pi_j} and $\bar{\ell} = \ell+2$. Otherwise, $(a,b)$ is the second-to-last edge on \stackpath{\ell}{\pi_i, \pi_j}, in which case $P$ equals \stackpath{\ell}{\pi_i, \pi_j} with the last edge removed and $\bar{\ell} = \ell+1$. Then, we define $\tilde{\pi} = \phi^{i,j}_{a,b}(\pi, \ell)$ by
\[
    \tilde{\pi}(a) = \pi(w_1), \quad \tilde{\pi}(b) = \pi(w_2), \quad \tilde{\pi}(w_i) = \pi(w_{i+2}) \quad \text{for $i \in [1, \bar{\ell}-2]$},
\]
\[
    \tilde{\pi}(v) = \pi(v) \quad \text{for } v \in V \setminus \{w_1, \dots, w_{\bar{\ell}}\}.
\]
See \Cref{fig:plus-rotation}.

\begin{figure}
    \centering
    \includegraphics[scale=0.9]{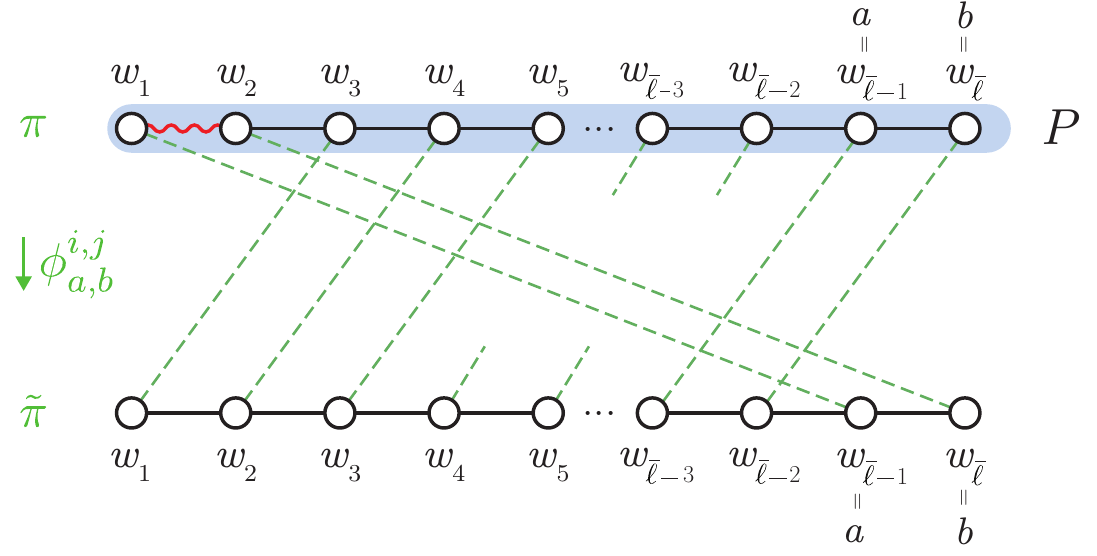}
    \caption{\small The permutation $\tilde{\pi} = \phi^{i,j}_{a,b}(\pi, \ell)$ resulting from rotating the ranks of the vertices on the prefix $P$ of the path \stackpath{\ell}{\pi_i, \pi_j} under $\pi$. The prefix $P$ is highlighted in blue. Dashed green lines connect pairs of vertices with the same rank. The ranks $\{i,j\}$ are rotated from vertices $\{w_1, w_2\}$ under $\pi$ to $\{a,b\}$ under $\tilde{\pi}$. The rank of any vertex not on $P$ is kept unchanged.}
    \label{fig:plus-rotation}
\end{figure}

We make the following claim on the map $\phi^{i,j}_{a,b}$, which will imply that for any $\tilde{\pi} \in U^{i,j}_{a,b}$, the preimage $(\phi^{i,j}_{a,b})^{-1}(\tilde{\pi})$ has size at most a small constant:

\begin{claim}\label{cl:mapinjectiveplus}
Suppose $\phi^{i,j}_{a,b}(\pi, \ell) = \phi^{i,j}_{a,b}(\pi', \ell')$ for two pairs $(\pi, \ell), (\pi', \ell') \in T^{i,j}(a,b)$. Let $P$ (resp. $P'$) be the prefix of the path \stackpath{\ell}{\pi_i, \pi_j} (resp.  \stackpath{\ell'}{\pi_i', \pi_j'}) ending at $(a,b)$. Then one of the following holds:
\begin{itemize}
    \item $P = P'$ and $\pi = \pi'$.
    
    \item One of $P, P'$ contains the other as a subpath with one fewer vertex.
\end{itemize}
\end{claim}

We defer the proof of \Cref{cl:mapinjectiveplus} to the end of this subsection, and first use it to finish proving \Cref{cl:boundpairplus}. We show that for any $\tilde{\pi} \in U^{i,j}_{a,b}$, the preimage $(\phi^{i,j}_{a,b})^{-1}(\tilde{\pi})$ has size at most $4$. This will imply that
\[
    |T^{i,j}(a,b)| \leq 4|U^{i,j}_{a,b}| = 8(n-2)!.
\]
We write
\[
   (\phi^{i,j}_{a,b})^{-1}(\tilde{\pi}) = \{(\pi^1, \ell_1), \dots, (\pi^m, \ell_m)\}
\]
and show that $m \leq 4$. For each $k = 1, \dots, m$, let $P^k$ denote the prefix of the path \stackpath{\ell_k}{\pi^k_i, \pi^k_j} ending at $(a,b)$. By \Cref{cl:mapinjectiveplus}, there are only two possibilities among the prefixes $P^1, \dots, P^m$. By permuting the indices, we assume that
\[  P^1 = \cdots = P^{m'} \neq P^{m'+1} = \cdots = P^m   \]
for some $0 \leq m' \leq m$. Next, we show that $m' \leq 2$. For the first $m'$ pairs in $(\phi^{i,j}_{a,b})^{-1}(\tilde{\pi})$, we have by \Cref{cl:mapinjectiveplus} that
\[
    \pi^1 = \cdots = \pi^{m'}.
\]
We denote this common permutation by $\pi$. Now by the definition of $T^{i,j}(a,b)$, $(a,b)$ appears as one of the last two edges on each of \stackpath{\ell_1}{\pi_i, \pi_j}, $\dots$, \stackpath{\ell_m}{\pi_i, \pi_j}. This implies that there are only two possibilities among $\ell_1, \dots, \ell_{m
'}$, i.e. $m' \leq 2$. Similarly, we have $m - m' \leq 2$. Thus, $m \leq 4$.
\end{proof}

\newcommand{\pivotspivotprime}[0]{\ensuremath{P_{\text{PIV}}'}}

\begin{proof}[Proof of \Cref{cl:mapinjectiveplus}]
Note that if $P = P'$, then $\phi^{i,j}_{a,b}(\pi, \ell) = \phi^{i,j}_{a,b}(\pi', \ell')$ would imply that $\pi = \pi'$. Thus in what follows, we assume that $P \neq P'$ and show that one must contain the other as a subpath with one fewer vertex.

We start by recalling some notations for $\pi$ and setting up their counterparts for $\pi'$. Let \pivotspivot{} (resp. \pivotspivotprime{}) denote the set of pivots found by \pivot{} under the permutation $\pi$ (resp. $\pi'$) (\Cref{def:pivotset}). For each vertex $v \in V$, let $p_v \in \pivotspivot{}$ (resp. $p_v' \in \pivotspivotprime{}$) denote the pivot of $v$ in \pivot{} under $\pi$ (resp. $\pi'$). (\Cref{def:pivot}).

Now, we write
\[
    P = (w_1, \dots, w_{\bar{\ell}-2}, w_{\bar{\ell}-1} = a, w_{\bar{\ell}} = b), \quad P' = (w_1', \dots, w_{\bar{\ell}'-2}', w_{\bar{\ell}'-1}' = a, w_{\bar{\ell}'}' = b).
\]
Note that $\{\pi(w_1), \pi(w_2)\} = \{\pi'(w_1'), \pi'(w_2')\} = \{i,j\}$. Let $m$ be the largest integer such that
\[
    (w_{\bar{\ell}-m+1}, \dots, w_{\bar{\ell}-1} = a, w_{\bar{\ell}} = b) = (w_{\bar{\ell}'-m+1}', \dots, w_{\bar{\ell}'-1}' = a, w_{\bar{\ell}'}' = b),
\]
which is the maximal common suffix of $P$ and $P'$. Then $m \geq 2$. Observe that $\phi^{i,j}_{a,b}(\pi, \ell) = \phi^{i,j}_{a,b}(\pi', \ell')$ directly implies the following properties:

\begin{observation}\label{obs:pair-paths-properties}
The following properties hold for the permutations $\pi, \pi'$, paths $P,P'$, and integer $m$:
\begin{enumerate}[label=\rm{(\alph*)}]
    \item $\pi(v) = \pi'(v)$ for any vertex $v \in V \setminus (\{w_1, \dots, w_{\bar{\ell}}\} \cup \{w_1', \dots, w_{\bar{\ell}'}'\})$. \label{item:pair-paths-properties-a}
    
    \item $\pi(w_1) = \pi'(w_1')$, $\pi(w_2) = \pi'(w_2')$.  \label{item:pair-paths-properties-b}
    
    \item $\pi(w_{\bar{\ell}-k}) = \pi'(w_{\bar{\ell}'-k})$ for any $0 \leq k \leq m-3$. \label{item:pair-paths-properties-c}
\end{enumerate}
\end{observation}

Now we show that one of $P, P'$ must contain the other as a subpath with one fewer vertex. We denote
\[
    y = w_{\bar{\ell}-m+1} = w_{\bar{\ell}'-m+1}', \quad z = w_{\bar{\ell}-m+2} = w_{\bar{\ell}'-m+2}'.
\]
That is, $y$ is the starting vertex of the maximal common suffix of $P$ and $P'$, and $z$ is the next vertex. Then there are only two possibilities for the relation between $P$ and $P'$:
\begin{enumerate}[label=\rm{(\Alph*)}]
    \item Neither of $P, P'$ contains the other as a subpath. We assume that $\pi(z)< \pi'(z)$. \label{item:case-plus-A}
    
    \item One of $P, P'$ contains the other as a subpath. We assume that $P$ contains $P'$. Then $m = \bar{\ell}' < \bar{\ell}$.  \label{item:case-plus-B}
\end{enumerate}
See \Cref{fig:plus-cases}. We show that Case \ref{item:case-plus-A} is impossible, and that in Case \ref{item:case-plus-B}, we must have $\bar{\ell} - \bar{\ell}' = 1$.

\begin{figure}[h]
    \centering
    \includegraphics[scale=0.85]{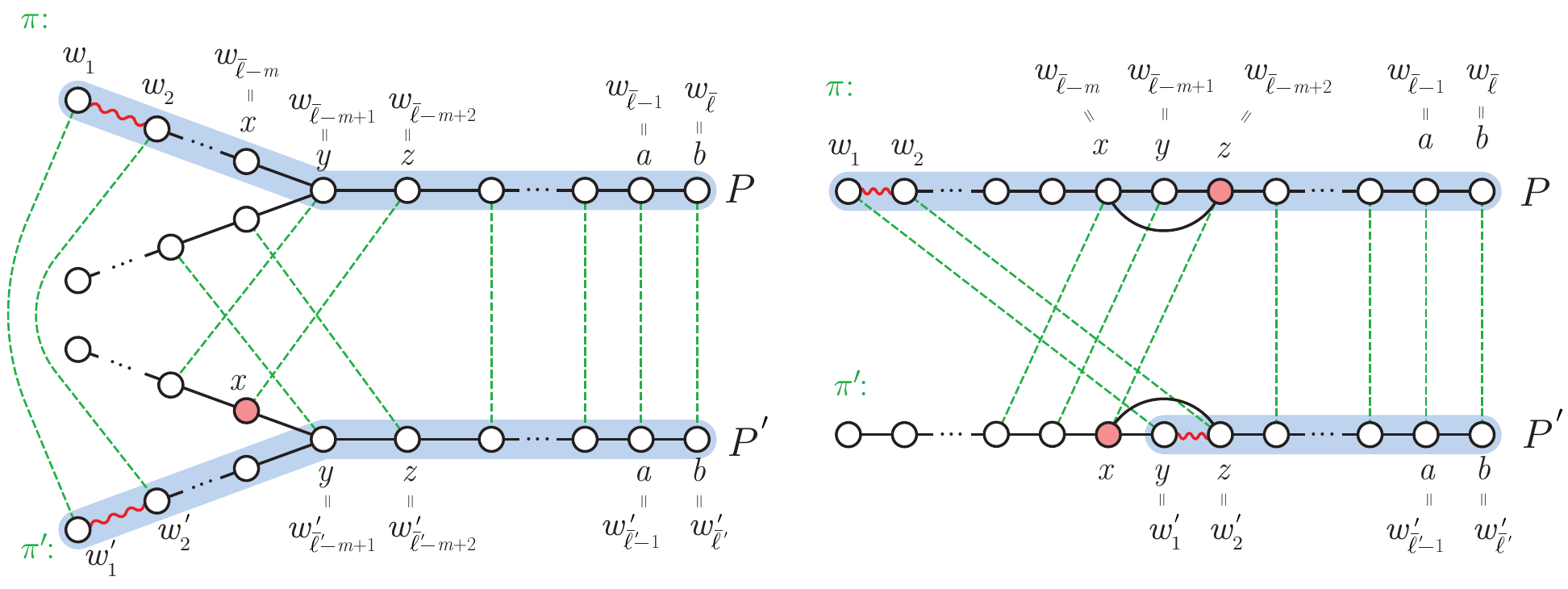}
    \caption{\small Cases \ref{item:case-plus-A} and \ref{item:case-plus-B}. In each case, the paths $P$ and $P'$ under permutations $\pi$ and $\pi'$ respectively are highlighted in blue. Dashed green lines connect pairs of vertices with the same rank. In both cases, the vertex $x = w_{\bar{\ell}-m}$ highlighted in red is the pivot of $y = w_{\bar{\ell}’-m+1}'$ in \pivot{} under $\pi'$. In Case \ref{item:case-plus-B}, where $m = \bar{\ell}'$, $x$ is the pivot of both $y = w_{1}'$ and $z = w_{2}'$ in \pivot{} under $\pi'$, and $z$ is a pivot in \pivot{} under $\pi$, also highlighted in red.}
    \label{fig:plus-cases}
\end{figure}

We note that $\pi(z) < \pi'(z)$ in Case \ref{item:case-plus-B} as well, since $\{y, z\} = \{\pi_i', \pi_j'\}$ in this case, which implies that $\pi'(z) \geq i > \pi(z)$ (recall $i<j$). As a consequence of \Cref{obs:pair-paths-properties} \ref{item:pair-paths-properties-a} \ref{item:pair-paths-properties-c}, we have 
\begin{equation}\label{eq:small-rank-same}
    \pi_k = \pi_k', \quad \text{for all $k < \pi(z)$.}
\end{equation}
This implies that, if $v \in V$ is a vertex such that $\pi(v) < \pi(z)$ or $\pi'(v) < \pi(z)$, then $\pi(v) = \pi'(v)$. Moreover, \FnVx{v} makes the same \emph{recursive} calls to \FnVx under $\pi$ and $\pi'$. In particular, $v$ is a pivot in \pivot{} under $\pi$ (i.e. $v \in \pivotspivot{}$) if and only if $v$ is a pivot under $\pi'$ (i.e. $v \in \pivotspivot{}'$). The following observation directly follows: 

\begin{observation}\label{obs:small-rank-same}
Let $v \in V$ be a vertex with pivot $p_v$ in \pivot{} under $\pi$ and $p_v'$ under $\pi'$. If $\pi(p_v) < \pi(z)$ or $\pi'(p_v')< \pi(z)$, then $p_v = p_v'$ and $\pi(p_v) = \pi'(p_v')$.
\end{observation}

We set $x = w_{\bar{\ell}-m}$ in both cases. Then
\[
    \pi'(x) = \pi(z).
\]
Note by \Cref{obs:query-path} that $y$ directly queries $z$ under $\pi$. Moreover, $x$ directly queries $y$ under $\pi$ unless $\{x,y\} = \{\pi_i, \pi_j\}$, in which case $\{x,y\}$ directly queries $z$.

We first claim that in both cases, $x$ is a pivot in \pivot{} under $\pi'$, i.e. $x \in \pivotspivot{}'$. Otherwise, $x \neq p_x'$ or equivalently $\pi'(p_x') < \pi'(x) = \pi(z)$. Then \Cref{obs:small-rank-same} implies that $p_x = p_x'$ and $\pi(p_x) = \pi'(p_x')$. Thus, $\pi(p_x) = \pi'(p_x') <\pi(z) < \pi(y)$. From \Cref{obs:vertex-query} \ref{item:query-vertex-nonpivot} \ref{item:query-pair-nonpivot}, this contradicts that $x$ directly queries $y$ under $\pi$ in the case $\{x,y\} \neq \{\pi_i, \pi_j\}$, or $\{x,y\}$ directly queries $z$ in the case $\{x,y\} = \{\pi_i, \pi_j\}$.

Next, we claim that in both cases, $x$ is the pivot of $y$ in \pivot{} under $\pi'$, i.e. $x = p_y'$. Otherwise, we have $\pi'(p_y')< \pi'(x) = \pi(z)$. By \Cref{obs:small-rank-same}, $p_y = p_y'$ and $\pi(p_y) = \pi'(p_y')$. Thus, $\pi(p_y) = \pi'(p_y') < \pi(z)$. From \Cref{obs:vertex-query} \ref{item:query-vertex-nonpivot}, this contradicts that $y$ directly queries $z$ under $\pi$.

At this point, we can already show that Case \ref{item:case-plus-A} is impossible, as follows. In this case, by \Cref{obs:query-path}, $y$ directly queries $z$ under $\pi'$. However, this contradicts that $\pi'(p_y') = \pi'(x) = \pi(z) < \pi'(z)$, again from \Cref{obs:vertex-query} \ref{item:query-vertex-nonpivot}.

It remains to complete the proof in Case \ref{item:case-plus-B}. In this case, since $\{y, z\} = \{\pi_i', \pi_j'\}$ is a pair in $X$ under the permutation $\pi'$ and $x = p_y'$, we also have $x = p_z'$. In particular, $\{x,z\} \in E$.

We now claim that $z$ is a pivot in \pivot{} under $\pi$, i.e. $z \in \pivotspivot{}$. Otherwise, $z \neq p_z$ or equivalently $\pi(p_z) < \pi(z)$. Then \Cref{obs:small-rank-same} implies that $p_z' = p_z$ and $\pi'(p_z') = \pi(p_z)$. But now, $\pi'(p_z') = \pi(p_z) < \pi(z) = \pi'(x)$. This contradicts that $x$ is the pivot in \pivot{} that $z$ joins under $\pi'$.

As a consequence, since $y$ directly queries $z$ under $\pi$, we have by \Cref{obs:vertex-query} \ref{item:query-vertex-nonpivot}
that $z$ must be the pivot of $y$ in \pivot{} under $\pi$, i.e. $z = p_y$. Moreover, since $z$ is a also a neighbor of $x$, we have $\pi(p_x) \leq \pi(z)$. Thus $\pi(y) > \pi(z) \geq \pi(p_x)$. By \Cref{obs:vertex-query} \ref{item:query-vertex-pivot} \ref{item:query-vertex-nonpivot}, $x$ cannot directly query $y$ under $\pi$. Then we must have $\{x,y\} = \{\pi_i, \pi_j\}$, i.e. $\{x, y\}$ is the first edge on $P$ (or \stackpath{\ell}{\pi_i, \pi_j}). This then implies that $\bar{\ell} = \bar{\ell}'+1$, as desired.
\end{proof}

\subsubsection{Proof of \Cref{lem:querypathminus}}\label{sec:proofminus}
Throughout this subsection, we fix distinct vertices $a, b \in V$ such that $\{a,b\} \not \in E$. As in \Cref{sec:proofplus}, given a permutation $\pi$ over $V$ and $i \in [n]$, we denote by $\pi_i \in V$ the vertex with rank $i$ under $\pi$.

Similar to the proof of \Cref{lem:querypathplus}, we break down the charges to bad triangles involving both $a$ and $b$ by specifying the  ranks of the pair of vertices in $X$ that initiates the charge. Given $i,j \in [n]$ with $i<j$, we define $T^{i,j}(a,b)$ to be the set of pairs $(\pi, \ell)$ of a permutation $\pi$ and an integer $\ell \in [2, 2r]$ such that under $\pi$, $\{\pi_i, \pi_j\} \in X$ and the last three vertices in the ordered list \stackpath{\ell}{\pi_i, \pi_j} are $(a,c,b)$ for some $c \in V$. We will prove the following bound on the size of $T^{i,j}(a,b)$, which is analogous to \Cref{cl:boundpairplus}:

\begin{claim}\label{cl:boundpairminus}
For any $i,j \in [n]$ with $i<j$ and any $r \geq 1$, we have
\[
    |T^{i,j}(a,b)| \leq 4(n-2)!.
\]
\end{claim}

\Cref{cl:boundpairminus} implies \Cref{lem:querypathminus} in the same way as how \Cref{cl:boundpairplus} implies \Cref{lem:querypathplus} (see \Cref{sec:proofplus}), and we omit the proof here.

\begin{proof}[Proof of \Cref{cl:boundpairminus}]
We define a map
\[
    \phi^{i,j}_{a,b}: T^{i,j}(a,b) \to U^{i,j}_{a,b},
\]
where
\[
    U^{i,j}_{a,b} = \{ \tilde{\pi} \mid \{\tilde{\pi}_i, \tilde{\pi}_j\} = \{a, b\}\}
\]
as in the proof of \Cref{cl:boundpairplus}. For $(\pi, \ell) \in T^{i,j}(a,b)$, the permutation $\phi^{i,j}_{a,b}(\pi, \ell)$ is defined by rotating the ranks of the vertices along the path $\stackpath{\ell}{\pi_i, \pi_j}$ while skipping the vertex between $a$ and $b$. Formally, we write
\[
    \stackpath{\ell}{\pi_i, \pi_j} = (w_1, w_2, \dots, w_{\ell} = a, w_{\ell+1}, w_{\ell+2} = b),
\]
where $\{w_1, w_2\} = \{\pi_i, \pi_j\}$. Then, we define $\tilde{\pi} = \phi^{i,j}_{a,b}(\pi, \ell)$ by
\[
    \tilde{\pi}(a) = \pi(w_1), \quad \tilde{\pi}(b) = \pi(w_2), \quad \tilde{\pi}(w_{\ell+1}) = \pi(w_{\ell+1}),
\]
\[
    \quad \tilde{\pi}(w_i) = \pi(w_{i+2}) \quad \text{for $i \in [1, \ell-2 ]$}, \quad \tilde{\pi}(w_{\ell-1}) = \pi(w_{\ell+2}), 
\]
\[
    \tilde{\pi}(v) = \pi(v) \quad \text{for } v \in V \setminus \{w_1, \dots, w_{\ell+2}\}.
\]
See \Cref{fig:minus-rotation}.

\begin{figure}[h]
    \includegraphics[scale=0.85]{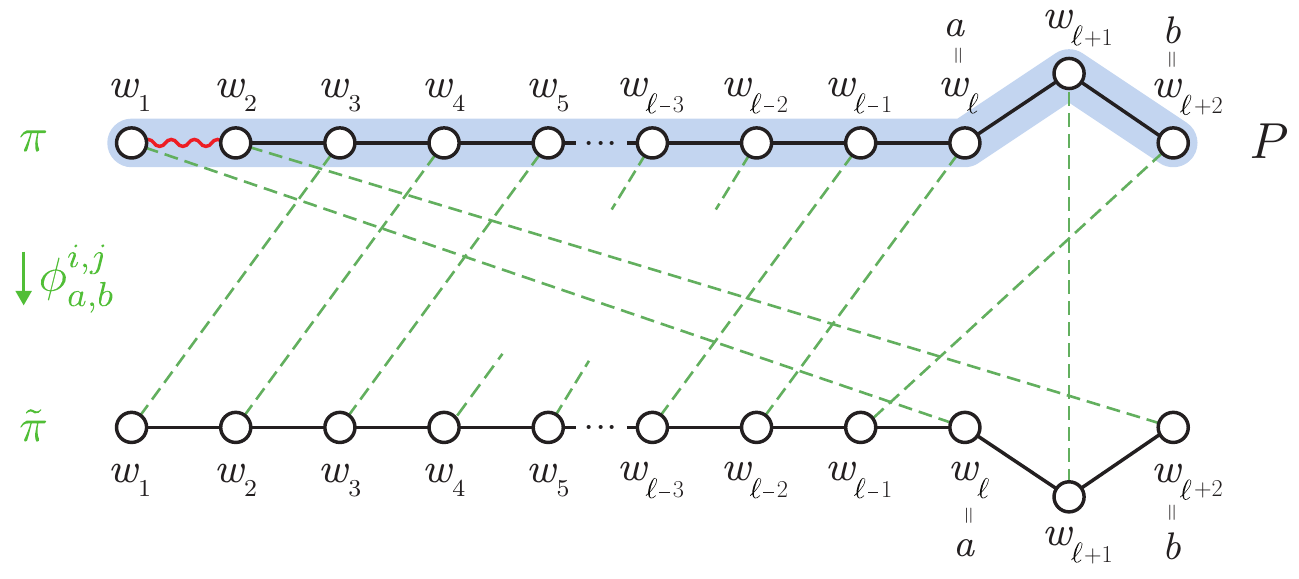}
    \centering
    \caption{\small The permutation $\tilde{\pi} = \phi^{i,j}_{a,b}(\pi, \ell)$ resulting from rotating the ranks of the vertices on the path \stackpath{\ell}{\pi_i, \pi_j} under $\pi$. The path \stackpath{\ell}{\pi_i, \pi_j} is highlighted in blue. Dashed green lines connect pairs of vertices with the same rank. The ranks $\{i,j\}$ are rotated from vertices $\{w_1, w_2\}$ under $\pi$ to $\{a,b\}$ under $\tilde{\pi}$. The rotation skips the vertex $w_{\ell+1}$ between $a$ and $b$ and leaves its rank unchanged. The rank of any vertex not on \stackpath{\ell}{\pi_i, \pi_j} is kept unchanged as well.}
    \label{fig:minus-rotation}
\end{figure}

We make the following claim on the map $\phi^{i,j}_{a,b}$, which directly implies that for any $\tilde{\pi} \in U^{i,j}_{a,b}$, the preimage $(\phi^{i,j}_{a,b})^{-1}(\tilde{\pi})$ has size at most 2:

\begin{claim}\label{cl:mapinjectiveminus}
Suppose $\phi^{i,j}_{a,b}(\pi, \ell) = \phi^{i,j}_{a,b}(\pi', \ell')$ for two distinct pairs $(\pi, \ell), (\pi', \ell') \in T^{i,j}(a,b)$. Then one of \stackpath{\ell}{\pi_i, \pi_j}, \stackpath{\ell'}{\pi_i', \pi_j'} contains the other as a subpath with one fewer vertex.
\end{claim}

    

It follows from \Cref{cl:mapinjectiveminus} that
\[
    |T^{i,j}(a,b)| \leq 2|U^{i,j}_{a,b}| = 4(n-2)!.
\]
We prove \Cref{cl:mapinjectiveminus} below. 
\end{proof}

\begin{proof}[Proof of \Cref{cl:mapinjectiveminus}]
Note that $\stackpath{\ell}{\pi_i, \pi_j} = \stackpath{\ell'}{\pi_i', \pi_j'}$ would imply that $\ell = \ell'$, and additionally that $\pi = \pi'$ since $\phi^{i,j}_{a,b}(\pi, \ell) = \phi^{i,j}_{a,b}(\pi', \ell')$. Thus we must have $\stackpath{\ell}{\pi_i, \pi_j} \neq \stackpath{\ell'}{\pi_i', \pi_j'}$.

We use the same notations for $\pi$ and $\pi'$ as in the proof of \Cref{cl:mapinjectiveplus}. We write
\[
    P = \stackpath{\ell}{\pi_i, \pi_j} = (w_1, \dots, w_{\ell} =a, w_{\ell+1} = c, w_{\ell+2} = b),
\]
\[
    P' = \stackpath{\ell'}{\pi_i', \pi_j'} =  (w_1', \dots, w_{\ell'}' =a, w_{\ell'+1}' = c', w_{\ell'+2}' = b).
\]
Note that $\{\pi(w_1), \pi(w_2)\} = \{\pi'(w_1'), \pi'(w_2')\} = \{i,j\}$. We consider three cases separately:
\begin{enumerate}[label=(\rm{\Roman*})]

    \item $c = c'$ and $w_{\ell-1} = w_{\ell'-1}'$. \label{item:proof-minus-case-i}
    
    \item $c = c'$ and $w_{\ell-1} \neq w_{\ell'-1}'$. \label{item:proof-minus-case-ii}
    
    \item $c \neq c'$.
    \label{item:proof-minus-case-iii}   
\end{enumerate}

In Cases \ref{item:proof-minus-case-i} or \ref{item:proof-minus-case-ii}, the two paths $P$ and $P'$ end at the same three vertices $(a,c,b)$. Let $m$ be the largest integer such that
\[
    (w_{\ell-m+3}, \dots, w_{\ell} = a, w_{\ell+1} = c, w_{\ell+2} = b) = (w_{\ell'-m+3}', \dots, w_{\ell'}' = a, w_{\ell'+1}' = c, w_{\ell'+2}' = b),
\]
which is the maximal common suffix of $P$ and $P'$. Then $m \geq 3$, and equality holds if and only if we are in Case \ref{item:proof-minus-case-ii}.

We first consider Case \ref{item:proof-minus-case-i}, where $m \geq 4$. Then, after setting $\bar{\ell} = \ell+2, \bar{\ell}' = \ell' +2$, we see that $\phi^{i,j}_{a,b}(\pi, \ell) = \phi^{i,j}_{a,b}(\pi', \ell')$ directly implies that the properties in \Cref{obs:pair-paths-properties} hold for the permutations $\pi, \pi'$, paths $P,P'$, and integer $m$. Therefore, we may finish Case \ref{item:proof-minus-case-i} by the same paragraphs after \Cref{obs:pair-paths-properties} up to the end of the proof of \Cref{cl:mapinjectiveplus}.

\begin{figure}[h]
    \centering
    \includegraphics[scale=0.9]{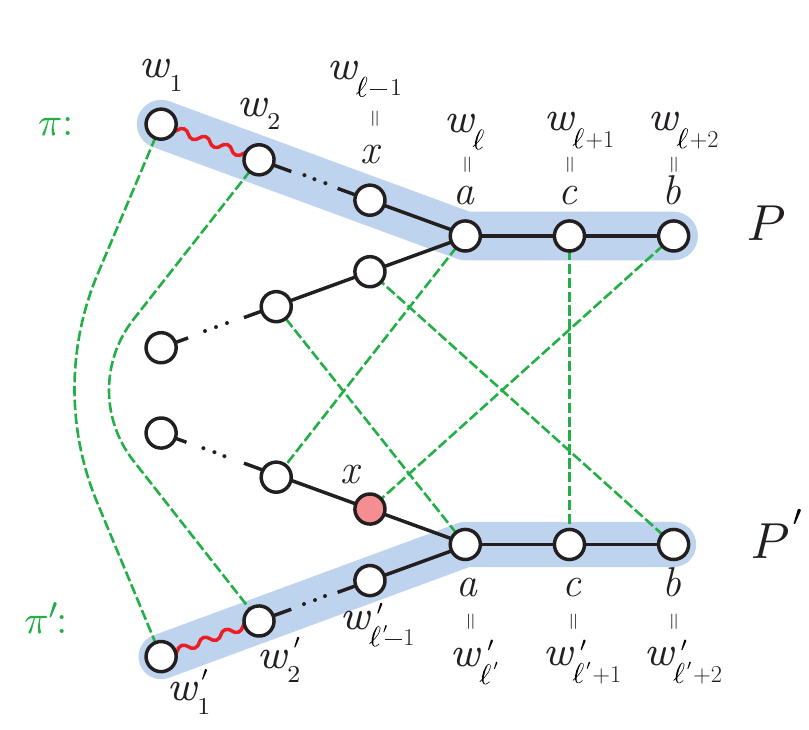}
    \caption{\small Case \ref{item:proof-minus-case-ii} where $m=3$. The paths $P$ and $P'$ under permutations $\pi$ and $\pi'$ respectively are highlighted in blue. Dashed green lines connect pairs of vertices with the same rank. The vertex $x = w_{\ell-1}$ highlighted in red is the pivot of $a$ in \pivot{} under $\pi'$.}
    \label{fig:minus-case-ii}
\end{figure}

Next, we show that Case \ref{item:proof-minus-case-ii}, where $m = 3$, is impossible. See \Cref{fig:minus-case-ii}. In this case, $\pi(c) = \pi'(c)$ but $\pi(b) \neq \pi'(b)$, and we assume without loss of generality that $\pi(b)< \pi'(b)$. From $\phi^{i,j}_{a,b}(\pi, \ell) = \phi^{i,j}_{a,b}(\pi', \ell')$, we directly have
\[
    \pi_k = \pi_k', \quad \text{for all $k < \pi(b)$,}
\]
which is similar to \Cref{eq:small-rank-same}. This implies the following observation, similar to \Cref{obs:small-rank-same}:

\begin{observation}\label{obs:small-rank-same-minus}
Let $v \in V$ be a vertex with pivot $p_v$ in \pivot{} under $\pi$ and $p_v'$ under $\pi'$. If $\pi(p_v) < \pi(b)$ or $\pi'(p_v')< \pi(b)$, then $p_v = p_v'$ and $\pi(p_v) = \pi'(p_v')$.
\end{observation}

We set $x = w_{\ell-1}$. Then
\[
    \pi'(x) = \pi(b).
\]
Based on \Cref{obs:small-rank-same-minus}, we can use the same argument as that after \Cref{obs:small-rank-same} in the proof of \Cref{cl:mapinjectiveplus} to deduce that under $\pi'$, $x$ is the pivot of $a$ in \pivot, i.e. $x = p_a' \in \pivotspivotprime$.  However, since $\pi'(c) = \pi(c) > \pi(b) =  \pi'(x)$, $a$ cannot directly query $c'$ under $\pi'$ by \Cref{obs:vertex-query} \ref{item:query-vertex-nonpivot}. This is a contradiction to \Cref{obs:query-path}. 

Finally, we show that Case \ref{item:proof-minus-case-iii} is impossible as well. In this case, $\pi(c) \neq \pi'(c')$, and we assume without loss of generality that $\pi(c) < \pi'(c')$. We will consider two subcases depending on whether $c$ is on the path $P'$ separately. See \Cref{fig:minus-case-iii}.

\begin{figure}[h]
    \centering
    \includegraphics[scale=0.9]{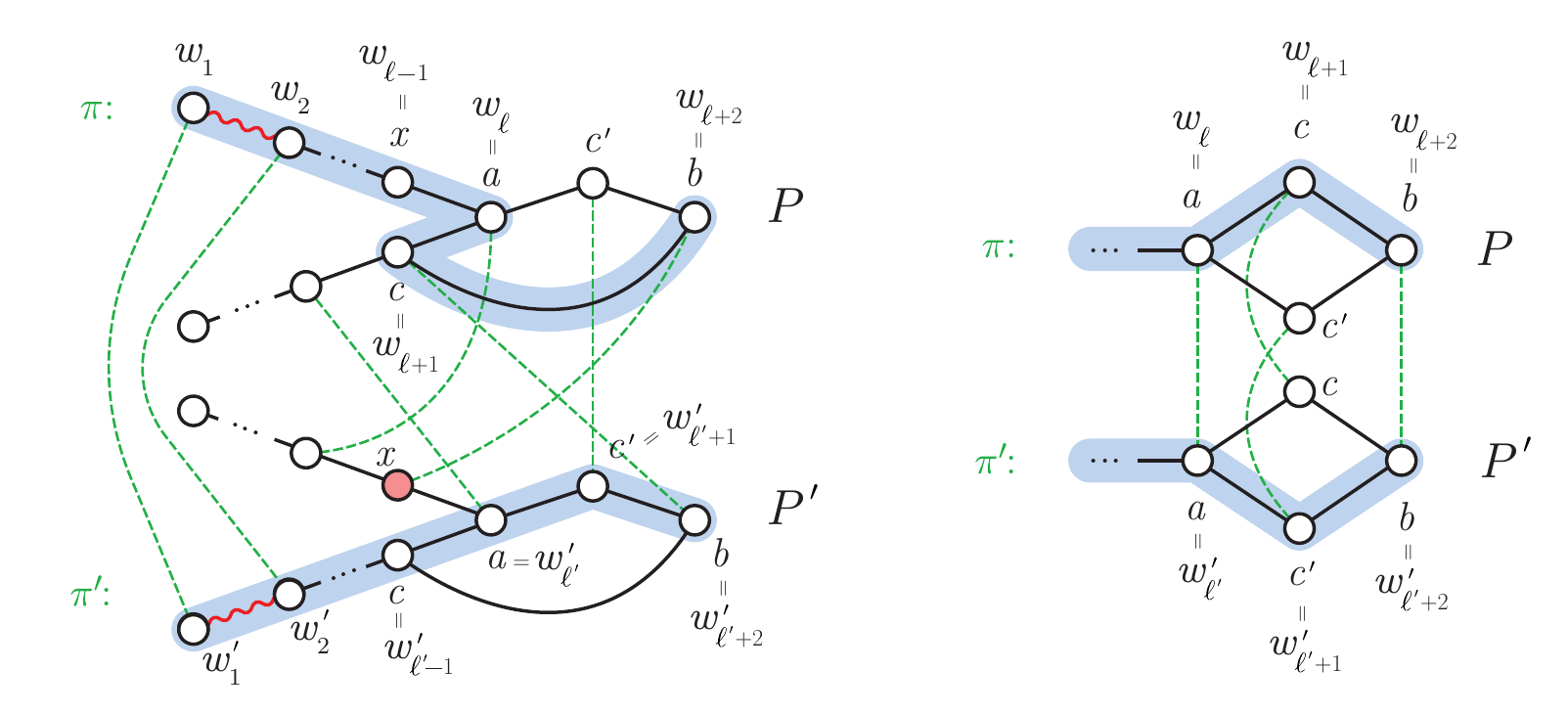}
    \caption{\small The two subcases of Case \ref{item:proof-minus-case-iii} where $c \neq c'$. In each subcase, the paths $P$ and $P'$ under permutations $\pi$ and $\pi'$ respectively are highlighted in blue. Dashed green lines connect pairs of vertices with the same rank. In the subcase where $c$ is on $P'$, the vertex $x = w_{\ell-1}$ highlighted in red is the pivot of $a$ in \pivot{} under $\pi'$.}
    \label{fig:minus-case-iii}
\end{figure}

First suppose that $c$ is on $P'$. By the definition of $\phi^{i,j}_{a,b}$, $\tilde{\pi}(c) = \pi'(v)$ for some vertex $v$ on $P'$. Now recall that $\tilde{\pi}(c) = \pi(c) < \pi'(c')$, which implies that $\pi'(v) < \pi'(c')$. However, the only vertex on $P'$ with rank less than $\pi'(c')$ under $\pi'$ is $b$, which means that $v = b$. Since $\pi'(b) = \tilde{\pi}(w_{\ell'-1}')$, we have $c = w_{\ell'-1}'$. This implies that
\[
    \pi(b) < \pi(c) = \tilde{\pi}(c) = \pi'(b) < \pi'(c').
\]
The proof from this point on is similar to Case \ref{item:proof-minus-case-ii} above. Here, $\phi^{i,j}_{a,b}(\pi, \ell) = \phi^{i,j}_{a,b}(\pi', \ell')$ also implies that
\[
    \pi_k = \pi_k', \quad \text{for all $k < \pi(b)$.}
\]
Thus \Cref{obs:small-rank-same-minus} also holds. We again set $x = w_{\ell-1}$. Then
\[
    \pi'(x) = \pi(b).
\]
In particular, $\pi'(x) = \pi(b) < \pi'(c')$ implies that $x \neq c'$. Based on \Cref{obs:small-rank-same-minus}, we can use the same argument as that after \Cref{obs:small-rank-same} in the proof of \Cref{cl:mapinjectiveplus} to deduce that $x = p_a' \in \pivotspivotprime$. However, since $\pi'(c')> \pi'(x)$, $a$ cannot directly query $c'$ under $\pi'$ by \Cref{obs:vertex-query} \ref{item:query-vertex-nonpivot}. This is a contradiction to \Cref{obs:query-path}. 

It remains to rule out the subcase where $c$ is not on $P'$. Here, we have
\[
    \pi'(c) = \tilde{\pi}(c) = \pi(c) < \pi'(c').
\]
By \Cref{obs:query-path}, under $\pi'$, $a$ directly queries $c'$. Thus by \Cref{obs:vertex-query} \ref{item:query-vertex-pivot} \ref{item:query-vertex-nonpivot}, $a$ also directly queries $c$ under $\pi'$ since it is a neighbor with rank smaller than $c'$. Now, if there exists a vertex $d \in V$ that $c$ directly queries under $\pi'$, then at some point, the stack of recursive calls to \FnVx when we call \FnPr{$w_1', w_2'$} consists of the following $\ell'$ elements
\[
    (w_3', \dots, w_{\ell'}' = a, c, d),
\]
and this happens before when the stack consists of $(w_3', \dots, w_{\ell'}' = a, c', b)$. This contradicts the definition of \stack{\ell'}{w_1',w_2'} and \stackpath{\ell'}{w_1', w_2'}. Otherwise, $c$ does not directly query any vertex under $\pi'$, which means that $c$ is a pivot in \pivot{}, i.e. $c \in \pivotspivotprime$. But this contradicts that $a$ directly queries the neighbor $c'$ whose rank is greater than that of $c$ under $\pi'$, by \Cref{obs:vertex-query} \ref{item:query-vertex-nonpivot}.
\end{proof}

\subsection{Deferred Proofs}\label{sec:deferred}

To show \Cref{cl:stack-gets-large}, we will make use of the following claim on the stack of recursive calls to \FnVx when we call \FnVx{$w$} for a vertex $w \in V$:

\begin{claim}\label{cl:stack-vertex}
    Let $w$ be a vertex that remains unsettled after $t$ rounds in \ourpivot{}. Then when we call \FnVx{$w$}, at some point the stack of recursive calls to \FnVx includes at least $2t$ elements (excluding $w$).
\end{claim}

\begin{myproof}
We proceed by induction on $t$. The base case $t = 0$ clearly holds. Now assume the statement for $t = k$, and let $w$ be a vertex that remains unsettled after $k+1$ rounds. We consider two cases depending on whether $w$ is a pivot in \pivotspivot{} separately:

We first consider the case where $w \in \pivotspivot{}$ is a pivot. By \Cref{obs:vertex-query} \ref{item:query-vertex-pivot}, the set of vertices that $w$ directly queries is
\[
    Q_w = \{z \in N(w) \mid \pi(z) < \pi(w)\},
\]
and that $Q_w$ does not contain any pivots in \pivotspivot{}. Without loss of generality, we may assume that there does not exist a vertex $z \in Q_w$ that remains unsettled after $k+1$ rounds, since otherwise to prove the claim for $w$ it suffices to prove it for $z$.

Next, we claim that there exists a vertex $z \in Q_w$ that remains unsettled after $k$ rounds. Otherwise, since all neighbors of $w$ with smaller rank under $\pi$ are settled after $k$ rounds, $w$ would be marked as a pivot and marked as settled in round $k+1$, a contradiction. 

As a consequence, $p_z$ must also remain unsettled after $k$ rounds. Since $z \not \in \pivotspivot{}$, we have $p_z \neq z$. Thus by \Cref{obs:vertex-query} \ref{item:query-vertex-nonpivot}, $z$ directly queries $p_z$. By the inductive hypothesis applied to $p_z$, when we call \FnVx{$p_z$}, at some point the stack of recursive calls to \FnVx includes $2k$ elements (excluding $p_z$). This implies the claim for $w$.

Now, we consider the other case where $w \not \in \pivotspivot{}$ is non-pivot. By \Cref{obs:vertex-query} \ref{item:query-vertex-nonpivot}, $w$ directly queries $p_w$. We claim that $p_w$ also remains unsettled after $k+1$ rounds. Otherwise $p_w$ would have been marked as a pivot by the end of round $k+1$ and $w$ would have been removed together with $p_w$, a contradiction. By the previous case, the claim holds for $p_w$, which then implies the claim for $w$.
\end{myproof}

\begin{myproof}[Proof of \Cref{cl:stack-gets-large} via \Cref{cl:stack-vertex}]
It suffices to prove the case $\ell = 2r$. We consider the two cases for the pair $\{u,v\}$ given in \Cref{lem:X-pair-classification} separately. In Case \ref{item:case-i}, $p_{\{u,v\}}$ remains unsettled after $r$ rounds in \ourpivot{}.  If $p_{\{u,v\}} \not \in \{u,v\}$, then by \Cref{obs:vertex-query} \ref{item:query-pair-nonpivot}, $\{u,v\}$ directly queries $p_{\{u,v\}}$, and we may thus conclude by \Cref{cl:stack-vertex} applied to $w = p_{\{u,v\}}$. If $p_{\{u,v\}} \in \{u,v\}$, then we may see from \Cref{obs:vertex-query} \ref{item:query-vertex-pivot} \ref{item:query-pair-pivot} that every vertex $z$ directly queried by $p_{\{u,v\}}$ is also directly queried by $\{u,v\}$. We again conclude by \Cref{cl:stack-vertex} applied to $w = p_{\{u,v\}}$.

In Case \ref{item:case-ii}, $w_{\{u,v\}} \in N(u) \cup N(v)$ is an unsettled vertex after $r$ rounds in \ourpivot{} with $\pi(w_{\{u,v\}}) < \pi(p_{\{u,v\}})$. Thus $\{u,v\}$ directly queries $w_{\{u,v\}}$ by \Cref{obs:vertex-query} \ref{item:query-pair-pivot} \ref{item:query-pair-nonpivot}. We may then conclude by \Cref{cl:stack-vertex} applied to $w = w_{\{u,v\}}$.
\end{myproof}

\section{Implementations}\label{sec:implementations}

In this section, we prove the implications of \cref{thm:main} in the models discussed. In each model, given $\epsilon>0$, we take $r = O(1/\epsilon)$ and provide an implementation of our \ourpivot{} algorithm. We introduce the following notation: For each $t \in [1, r]$, and each vertex $v \in V$ that is unsettled at the beginning of round $t$ or \ourpivot{}, let $\eta_t(v)$ denote the vertex that has the smallest rank under $\pi$ among the vertices in $\{v\} \cup N(v)$ that are unsettled at the beginning of round $t$. Note that $v$ is marked as a pivot in round $t$ if and only if $\eta_t(v) = v$.

\paragraph{The Massively Parallel Computations (\MPC{}) Model:} In the \MPC{} model, the edge set $E$ of the input graph $G$ is distributed to a collection of machines. Computation then proceeds in synchronous rounds. In each round, a machine can receive messages from other machines in the previous round, perform some local computation, and send messages to other machines as input for the next round. Each message has size $O(1)$ words. Each machine has limited local space, which restricts the total number of messages it can receive or send in a round. For correlation clustering, at the end of the computation, each machine is required to know the cluster IDs of the vertices for which it initially holds edges. We focus on the \emph{strictly sublinear regime} of \MPC{} where the computation uses $O(n^\delta)$ space per machine, where $\delta>0$ is a constant that can be made arbitrarily small, and $O(m)$ total space.

We now show that there is a randomized $O(r)$-round MPC algorithm that obtains an expected $\big(3+\frac{8}{2r-1}+n^{-\Omega(1)}\big)$-approximation of correlation clustering and uses $O(n^\delta)$ space per machine, where constant $\delta > 0$ can be made arbitrarily small, and $O(m)$ total space.

\newcommand{\ourpivotvar}[0]{\ensuremath{r\textsc{-Pivot-Variant}}}

\begin{proof}[Proof of \Cref{cor:MPC}]
Since a random permutation of $V$ cannot be drawn and stored on a single machine, we implement the following variant of \ourpivot{} in the \MPC{} model:

\begin{whitetbox}
\textbf{Algorithm} \ourpivotvar{} \textbf{:} A variant of the algorithm \ourpivot{}.

\vspace{0.2cm}

\begin{itemize}[topsep=0pt, itemsep=0pt,leftmargin=13pt]
\item Instead of drawing a random permutation $\pi$ of the vertex set $V$, each vertex $v \in V$ draws a rank $\pi(v)$ from $\{0, \dots, n^c-1\}$ independently and uniformly at random, where $c$ is a sufficiently large constant.

\item Proceed in the same way as in \ourpivot{}, breaking ties in ranks in favor of the vertex with smaller ID.

\end{itemize}
\end{whitetbox}

\begin{claim}
For any $r \geq 1$, the algorithm \ourpivotvar{} outputs an expected $\big(3+\frac{8}{2r-1}+n^{-\Omega(1)}\big)$-approximation of correlation clustering.
\end{claim}

\begin{proof}
First, observe that any connected component of $G$ that is a clique can be successfully identified by both \ourpivot{} and \ourpivotvar{} with probability 1. Thus, if all connected components of $G$ are cliques, then both \ourpivot{} and \ourpivotvar{} output the optimal solution with probability 1.

Now consider the case where at least one connected component of $G$ is not a clique. Then $\opt{G} \geq 1$. With probability at least $1-n^{2-c}$, the ranks $\pi(v)$ drawn in \ourpivotvar{} are all distinct, in which case $\pi$ is equivalent to a uniformly drawn permutation of $V$, and the output of \ourpivotvar{} is the same as the output of \ourpivot{}, which has cost at most $\big(3+\frac{8}{2r-1}\big) \cdot \opt{G}$ by \Cref{cor:3-apx}. In the case there are repeated ranks, which happens with probability at most $n^{2-c}$, we simply upper bound the cost by ${n \choose 2} \leq n^2 \cdot \opt{G}$, since $\opt{G} \geq 1$. Thus the expected cost of the above variant is at most
$$
\left( 3+\frac{8}{2r-1} + n^{2-c} \cdot n^2\right) \opt{G} \leq \left( 3+\frac{8}{2r-1} + n^{-\Omega(1)} \right) \opt{G}.\qedhere
$$
\end{proof}

Now we describe the \MPC{} implementation. Take $r = O(1/\epsilon)$ such that $\frac{8}{2r-1} + n^{-\Omega(1)}< \epsilon$, and fix $\delta >0$. We use a collection of $O(n^{1-\delta})$ machines to draw the rank $\pi(v)$ of each vertex $v \in V$ from $\{0, \dots, n^c-1\}$ independently and uniformly at random, where $c$ is a sufficiently large constant. Let $M_v$ denote the machine that holds the rank $\pi(v)$ of a vertex $v$. We will also let $M_v$ keep track of whether $v$ is settled or a pivot throughout the computation. All vertices are initially marked as unsettled and non-pivot.

Moreover, the input edges in $E$ are stored in a collection of $O(m/n^{\delta})$ machines. In the implementation, we will need to perform the following two operations:
\begin{itemize}
    \item Every machine $M_v$ informs each machine that holds an input edge incident to $v$ whether $v$ is settled or a pivot.

    \item Every machine that holds input edges requests to mark some of the vertices it holds edges to as settled,  by communicating with the corresponding machines $M_v$'s.
    
\end{itemize}
We note that each operation can be done in $O(1/\delta)$ rounds of \MPC{} with $O(n^\delta)$ space per machine and $O(m)$ total space.

For each $t \in [1,r]$, we implement round $t$ of \ourpivot{} by $O(1/\delta)$ rounds of \MPC{} as follows: First, we compute $\eta_t(v)$ for all unsettled vertices $v \in V$ and store them in the respective machines $M_v$'s. To do this, we construct a set $L_t$ of at most $2m+n$ \emph{ordered} pairs of vertices as follows:
\begin{itemize}
    \item For each edge $\{u,v\} \in E$ such that both $u, v$ are unsettled, add both $(u,v)$ and $(v,u)$ to $L_t$. This is done by the machine holding the input edge $\{u,v\}$, which communicates with $M_u$ and $M_v$ to check whether $u,v$ are both unsettled. 
    
    \item For each unsettled vertex $v \in V$, add $(v,v)$ to $L_t$. This is done by the machine $M_v$.
\end{itemize}
Next, we sort the pairs in $L_t$ by the rank of the first vertex, and in case of ties, by the rank of the second vertex. This can be done in $O(1/\delta)$ rounds of \MPC{} with $O(n^\delta)$ space per machine and $O(m)$ total space \cite{Goodrich11}. The sorted list $L_t$ has form
\[ (u_1, v_{1,1}), \dots, (u_1, v_{1, \deg_t(u_1)+1}), (u_2, v_{2,1}), \dots, (u_2, v_{2, \deg_t(u_2)+1}), \dots, (u_{n_t}, v_{n_t,1}), \dots, (u_{n_t}, v_{n_t, \deg_t(u_{n_t})+1}),  \]
where $u_1, \dots, u_{n_t}$ is a listing of the unsettled vertices in increasing order of their ranks, and for each $i$, $v_{i,1}, \dots, v_{i, \deg_t(u_i)+1}$ is a listing of the unsettled vertices in $\{u_i\} \cup N(u_i)$ in increasing order of their ranks. In particular, for each $i$, $v_{i,1} = \eta_t(u_i)$. Then, any machine that holds $(u_i, v_{i,1})$ sends $v_{i,1} = \eta_t(u_i)$ to $M_{u_i}$.

Now, for any unsettled vertex $v$, if $M_v$ sees that $\eta_t(v) = v$, it marks $v$ as settled and a pivot. Then, any machine holding an input edge $\{u,v\} \in E$ checks whether $u$ is a newly-identified pivot and $v$ is unsettled, in which case it asks $M_v$ to mark $v$ as settled. Thereby we have implemented round $t$ of \ourpivot{} by $O(1/\delta)$ rounds of \MPC{}.

Now that we have implemented all $r$ rounds of \ourpivot{} by $O(r/\delta)$ rounds of \MPC{}, in $O(1/\delta)$ additional rounds, each $M_v$ determines the cluster ID of $v$, as follows: Every pivot starts a cluster which includes itself. Then, any non-pivot vertex $u$ can determine the vertex $v$ that has the smallest rank under $\pi$ among all pivots in $N(u)$ (if any), as well as the vertex $w$ that has the smallest rank under $\pi$ among all unsettled vertices in $N(u)$ (if any). The values $v$ and $w$ can be computed and sent to $M_u$ in a similar way as how the values $\eta_t(v)$'s are computed and sent above. If $v$ doesn't exist, or both $v, w$ exist and $\pi(w) < \pi(v)$, $u$ forms a singleton cluster. Otherwise, $u$ joins the cluster of $v$.  In the end, the cluster IDs are broadcast to the machines which initially hold input edges.
\end{proof}

\paragraph{The Graph Streaming Model:} In the streaming correlation clustering problem, the edges of graph $G$ arrive one by one in a stream, and in an arbitrary order. The algorithm has to take few passes over the input, use a small space, and output the clusters at the end. It is not hard to see that $\Omega(n)$ words of space are needed just to store the final output. 

We now show that there is a randomized $(2r+1)$-pass streaming algorithm using $O(n \log n)$ bits of space that obtains an expected $\big(3+\frac{8}{2r-1}\big)$-approximation of correlation clustering.

\begin{proof}[Proof of \Cref{cor:streaming}]
We take $r = O(1/\epsilon)$ such that $\frac{8}{2r-1}< \epsilon$ and implement \ourpivot{} in the streaming model. We start by drawing a random permutation $\pi$ of $V$ and marking all vertices as unsettled and non-pivot. Then, for each $t \in [1, r]$, we implement round $t$ of \ourpivot{} by making two passes of the stream. In the first pass, we maintain $\eta_t(v)$ associated to each unsettled vertex $v \in V$, which is initialized to $v$, as follows: whenever an edge $\{u,v\} \in E$ arrives with both $u,v$ unsettled, if $\pi(u) < \pi(\eta_t(v))$, we update $\eta_t(v)$ to be $u$; similarly, if $\pi(v) < \pi(\eta_t(u))$, we update $\eta_t(u)$ to be $v$. After this pass, any unsettled vertex $v$ with $\eta_t(v) = v$ is marked as settled and a pivot. In the second pass, we mark all neighbors of the newly-identified pivots as settled, as follows: whenever an edge $\{u,v\} \in E$ arrives with one of $u,v$ marked as a pivot and the other unsettled, we mark the unsettled vertex as settled.

Now that we have implemented all $r$ rounds of \ourpivot{} by $2r$ passes of the stream, we determine the output clustering by making an additional final pass. Every pivot starts a cluster which includes itself. During the final pass, for each non-pivot vertex $u$, we find the vertex $v$ that has the smallest rank under $\pi$ among all pivots in $N(u)$ (if any), as well as the vertex $w$ that has the smallest rank under $\pi$ among all unsettled vertices in $N(u)$ (if any). If $v$ doesn't exist, or both $v,w$ exist and $\pi(w) < \pi(v)$, $u$ forms a singleton cluster. Otherwise, $u$ joins the cluster of $v$.
\end{proof}

\paragraph{The Local Model:} In the local model, each vertex of the input graph $G$ hosts a computationally unbounded processor, computation proceeds in rounds, and adjacent vertices can exchange messages of any size in each round. The main question is the number of rounds it takes to solve a graph problem where the input is the same as the communication network $G$.


We now show that there is a randomized $(2r+1)$-round local algorithm that obtains an expected $\big(3+\frac{8}{2r-1}+n^{-\Omega(1)}\big)$-approximation of correlation clustering using $O(\log n)$-bit messages.

\begin{proof}[Proof of \Cref{cor:local}]
Take $r = O(1/\epsilon)$ such that $\frac{8}{2r-1}+n^{-\Omega(1)}< \epsilon$. Since we cannot globally draw a random permutation of $V$ in the local model, we implement \ourpivotvar{} given in the \MPC{} implementation (see the proof of \Cref{cor:MPC} above). To start, each vertex $v \in V$ draws its rank $\pi(v)$ from $\{0, \dots, n^c-1\}$ independently and uniformly at random, where $c$ is a sufficiently large constant, and gets marked as unsettled and non-pivot.  Then, for each $t \in [1, r]$, we implement round $t$ of \ourpivot{} by two rounds in the local model. In the first round, any unsettled vertex sends its rank to all of its neighbors. Then, any unsettled vertex $v$ can determine $\eta_t(v)$ from the ranks of all of its unsettled neighbors, and if $\eta_t(v) = v$, it gets marked as settled and a pivot. In the second round, any newly-identified pivot informs all of its neighbors. Then, any unsettled vertex that sees a newly-identified pivot neighbor gets marked as settled.

Now that we have implemented all $r$ rounds of \ourpivot{} by $2r$ rounds in the local model, we determine the output clustering by an additional final round. Every pivot starts a cluster which includes itself. In the final round, any pivot informs all of its neighbors that it is pivot and sends over its rank under $\pi$. Moreover, any unsettled vertex informs all of its neighbors that it is unsettled and sends over its rank under $\pi$. Then, any non-pivot vertex $u$ can determine the vertex $v$ that has the smallest rank under $\pi$ among all pivots in $N(u)$ (if any), as well as the vertex $w$ that has the smallest rank under $\pi$ among all unsettled vertices in $N(u)$ (if any). If $v$ doesn't exist, or both $v,w$ exist and $\pi(w) < \pi(v)$, $u$ forms a singleton cluster. Otherwise, $u$ joins the cluster of $v$.
\end{proof}

\paragraph{Local Computation Algorithms:} Centralized Local Computation Algorithms (LCAs) are useful for problems where both the input and output are too large. An LCA  is not required to output the whole solution, but instead should answer queries about parts of the output. For example, for the correlation clustering problem, the query to an LCA is a vertex $v$ and the answer should return the cluster ID of $v$. For graph problems, it is common to assume that the algorithm has query access to the adjacency lists. That is, for any vertex $v$, the LCA can query the degree of $v$ in $G$, and for any $i \in [\deg(v)]$, the LCA can request the ID of $i$-th neighbor of $v$. The goal is to answer each query using small time and space.

As it is by now standard and because \ourpivot{} can be implemented in $O(r)$ rounds in the local model, we can obtain a $\Delta^{O(r)}\poly\log n$ time/space LCA implementation for \ourpivot{} by simply collecting the whole $O(r)$-hop of any vertex \cite{RubinfeldTVX11,AlonRVX12}. This shows \Cref{cor:LCA}, with a choice of $r = O(1/\epsilon)$ such that $\frac{8}{2r-1}+n^{-\Omega(1)}< \epsilon$.

\section{Conclusion \& Open Problems}\label{sec:conclusion}

In this work, we showed that a $(3+\epsilon)$-approximation of correlation clustering can be obtained in $O(1/\epsilon)$ rounds in models such as (strictly sublinear) MPC, local, and streaming. This is a culminating point for low-depth algorithms for correlation clustering as the approximation gets close to a barrier of 3 for combinatorial algorithms and the round-complexity is essentially constant.

Several interesting questions remain open, especially in big data settings where there is no direct notion of round complexity. For example, is it possible to obtain an (almost) 3-approximation of correlation clustering in:
\begin{itemize}[itemsep=0pt]
    \item Sublinear time? See \cite{AssadiChen-ITCS} for the formal model for correlation clustering.
    \item $O(1)$ update-time in the fully dynamic model? See \cite{BehnezhadDHSS-FOCS19} for a $\poly\log n$ update-time algorithm. 
    \item $O(\Delta)$ time of the LCA model?
    \item A single pass of the streaming setting using $O(n \poly\log(n))$ space?
\end{itemize}

Finally, recall from the last paragraph of \cref{sec:ourcontribution} that our algorithm combined with the algorithm of \cite{ChawlaMSY14} gives an efficient $O(1/\epsilon)$-round algorithm for rounding the natural LP solution up to a factor of $(2.06 + \epsilon)$. Unfortunately, solving the LP remains the main bottleneck. It would thus be extremely interesting to study whether a low-round algorithm exists for solving the natural correlation clustering LP (see \cite{ChawlaMSY14}).




\bibliographystyle{plainnat}
\bibliography{references}

\appendix

\section{Size of {\normalfont \pivotsourpivot{}} vs {\normalfont \pivotspivot{}}}\label{sec:ourpiv-size}

In this section, we show that even though, by \cref{thm:main}, algorithm \ourpivot{} is almost as good as \pivot{} in terms of approximating correlation clustering, the set \pivotsourpivot{} of the pivots found by \ourpivot{} tends to be significantly smaller than the set \pivotspivot{} of the pivots found by \pivot{}. 

\begin{lemma}\label{lem:sizeofpivotsourpiv}
    For any $r \geq 1$, there is an infinite family of $n$-vertex graphs for which
    $$
        \E|\pivotsourpivot{}| \leq r \cdot n^{-\Omega(1/r)} \cdot \E|\pivotspivot{}|.
    $$
\end{lemma}

\newcommand{\degreeright}[1]{\ensuremath{d^\rightarrow_{#1}}}
\newcommand{\degreeleft}[1]{\ensuremath{d^\leftarrow_{#1}}}

\begin{proof}
We first construct a layered graph $H$, then the graph $G$ of the lemma will be the line-graph of $H$. That is, $G$ includes a vertex for each edge of $H$, and two vertices of $G$ are adjacent if their corresponding edges in $H$ share an endpoint. Observe that \pivotspivot{} for $G$ corresponds to a maximal matching of $H$, and \pivotsourpivot{} for $G$ corresponds to a (not necessarily maximal) matching of $H$.

Let $t := 2r+2$ where $r$ is the parameter we run \ourpivot{} with. Let $N$ be a sufficiently large even integer controlling the number of vertices in $H$ and let $\alpha = N^{1/3t} \geq 2$. The graph $H$ has $t$ layers of vertices $V_1, \ldots, V_t$ where $|V_i| = N/\alpha^{t-i}$. The graph has a perfect matching among the vertices in $V_t$. Additionally, there is a bipartite graph between every two consecutive layers. For any $i \in [t]$ define:
$$
    \degreeleft{i} := \alpha^{2(t-i)+1}, \qquad\qquad \degreeright{i} := \alpha^{2(t-i)}.
$$
 For any $i \in [t-1]$, each vertex $v \in V_i$ has exactly $\degreeright{i}$ neighbors in $V_{i+1}$, and each vertex in $V_{i+1}$ has exactly $\degreeleft{i+1}$ neighbors in $V_i$. For this to be doable we need to have $|V_i| \cdot \degreeright{i} = |V_{i+1}| \cdot \degreeleft{i+1}$ and $|V_1| \geq \degreeleft{2}$. The former holds because
$$
    |V_i| \cdot \degreeright{i} = \frac{N}{\alpha^{t-i}} \cdot \alpha^{2(t-i)} = \frac{N}{\alpha^{t-(i+1)}} \cdot \alpha^{2(t-(i+1))+1} = |V_{i+1}| \cdot \degreeleft{i+1},
$$
and the latter holds because 
$$
|V_1| = N/\alpha^{t-1} = N/N^{\frac{t-1}{3t}} = N^{2/3+1/3t} \qquad \& \qquad \degreeleft{2} = \alpha^{2(t-2)+1} = N^{\frac{2(t-2)+1}{3t}} \leq N^{2/3}.
$$

Since \pivotspivot{} is a maximal matching of $H$, it should be at least half the size of its maximum matching. Our construction has a perfect matching among the vertices of $V_t$. Thus, in any outcome of the \pivot{} algorithm we must have
$$
    |\pivotspivot{}| \geq \frac{1}{2} \cdot \frac{|V_t|}{2} = \Omega(N).
$$

Next we show that $\E|\pivotsourpivot{}| = O(tN/\alpha) = rN^{1-\Omega(1/r)}$, which proves the claim.

Let $\pi$ be a random rank over the {\em edges} of $H$ (or equivalently the vertices of $G$). Given a vertex $v_t \in V_t$, define $Q(v_t)$ to be the event of having a path $(v_t, v_{t-1}, \ldots, v_1)$ in $H$ where $v_i \in V_i$ and additionally $(v_i, v_{i-1})$ has the lowest rank in $\pi$ among all edges of $v_i$. Observe that under this event, all edges of $v_t$ remain unsettled by \ourpivot{} after  $(t-2)/2 = r$ rounds.

Let us lower bound $\Pr[Q(v_t)]$ for fixed $v_t \in V_t$. First, vertex $v_t$ has $\degreeleft{t}$ edges in $V_{t-1}$ and only one edge in $V_t$. So  the lowest rank edge of $v$ is to some vertex $v_{t-1} \in V_{t-1}$  with probability 
$$\frac{\degreeleft{t}}{\deg(v_t)} = \frac{\degreeleft{t}}{\degreeleft{t}+1} = \frac{\alpha}{\alpha+1}.$$
Conditioned on this, the lowest rank edge of $v_{t-1}$ goes to $v_{t-2} \in V_{t-2}$ with probability 
$$
\frac{\degreeleft{t-1}}{\deg(v_{t-1})+\deg(v_t)-1} = \frac{\degreeleft{t-1}}{\degreeleft{t-1}+\degreeright{t-1}+\degreeleft{t}} = \frac{\alpha^3}{\alpha^3+\alpha^2+\alpha} \geq \frac{\alpha}{\alpha+2},
$$
(Where here we are using the fact that for disjoint sets of edges $A$ and $B$, the lowest rank edge among $B$ has a lower rank than the lowest rank edge of $A$ with probability $\frac{|B|}{|B|+|A|}$.) Conditioned on this, the lowest rank edge of $v_{t-2}$ goes to $v_{t-3} \in V_{t-3}$ with probability 
$$
\frac{\degreeleft{t-2}}{\deg(v_{t-2})+\deg(v_{t-1})+\deg(v_t)-2} = \frac{\alpha^5}{\alpha^5+\alpha^4+\alpha^3+\alpha^2+\alpha - 1} \geq \frac{\alpha}{\alpha+2}.
$$
Continuing this argument all the way to $V_1$, we get
$$
    \Pr[Q(v_t)] \geq \left(\frac{\alpha}{\alpha+2}\right)^{t-1} = \left(1-\frac{2}{\alpha+2}\right)^{t-1} \geq 1 - \frac{2(t-1)}{\alpha+2} = 1 - \Theta(t/\alpha).\qedhere
$$
\end{proof}

\section{Deterministic Approximation of {\normalfont \ourpivot{}}}\label{sec:det-apx-ourpiv}

In this section, we prove the following:

\begin{lemma}\label{lem:worstcasepermutation}
    For any $r < n/2$, there exists a graph $G$ and a choice of $\pi$ for \ourpivot{} where 
    $$
    \costourpiv{G, \pi} = \Omega\left(\frac{n^2}{r} \cdot \opt{G}\right).
    $$
\end{lemma}
\begin{proof}
For some parameter $N$, let $G$ include a clique $K_N$ and a path of length $2r$ attached to the clique. Let $\pi$ be such that the vertices on the path have the lowest ranks in the graph and in a monotone increasing way from the degree one vertex of the path to its vertex attached to the clique. The ranks of the vertices in the clique can be arbitrary. An example construction for $r=3$ and $N = 8$ is shown below, with the numbers on the vertices corresponding to the ranks in $\pi$.
\begin{figure}[h]
    \centering
    \includegraphics{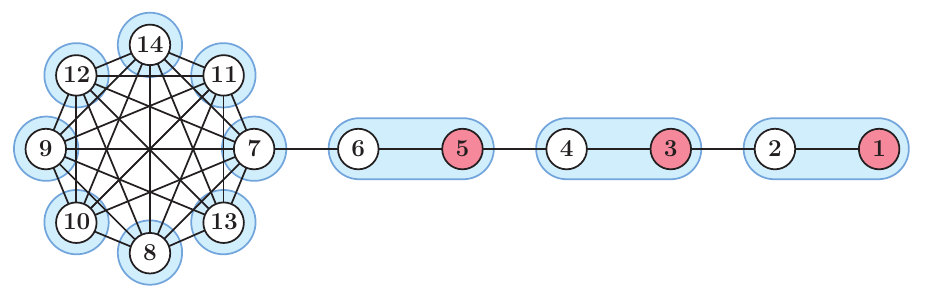}
\end{figure}

The idea is that within the first $r$ rounds of \ourpivot{}, only the vertices of the path will be marked as pivots. Thus, all vertices of the clique will form singleton clusters. As a result, the cost of this clustering is at least $\Omega(N^2)$. On the other hand, by simply putting all the vertices of the clique in a cluster, the optimal solution only pays a cost of $O(r)$. Thus the cost of \ourpivot{} is $\Omega(N^2/r)$ times that of the optimal solution.
\end{proof}




    

\section{Proof of \Cref{lem:width-gives-apx}}\label{sec:width-proof}
We provide a proof of \Cref{lem:width-gives-apx} based on the following result:

\begin{claim}[\cite{AilonCN-JACM08}]\label{cl:fractional-packing}
Let $BT$ be the set of all bad triangles in $G$. Let $y: BT \to [0, 1]$ be such that $\sum_{t \in BT: a, b \in t} y(t) \leq 1$ for every distinct pair of vertices $a,b \in V$ (not necessarily connected in $G$).  Then $\opt{G} \geq \sum_{t \in BT} y(t)$.
\end{claim}

\begin{myproof}
Consider the following LP, which clearly lower bounds the optimal solution since at least one pair of any bad triangle must be clustered incorrectly:
\[
	\begin{aligned}
		\text{minimize} \qquad & \sum_{a \neq b \in V} x_{\{a, b\}} && \\
		\text{subject to}\qquad & x_{\{a, b\}} + x_{\{a, c\}} + x_{\{b, c\}} \geq 1 \qquad &&\forall \{a, b, c\} \in BT\\
		& x_{\{a, b\}} \geq 0 \qquad &&\forall a \neq b \in V \quad
	\end{aligned}
\]
	Now consider the dual:
\[
	\begin{aligned}
		\text{maximize} \qquad &\sum_{t \in BT} y_t &&\\
		\text{subject to} \qquad &\sum_{t \in BT: a, b \in t } y_t \leq 1 \qquad &&\forall a \neq b \in V\\
		& y_t \geq 0 \quad &&\forall t \in BT
	\end{aligned}
\]
	The lemma then follows because the objective value of any feasible solution to the dual lower bounds the primal objective.
\end{myproof}

\begin{myproof}[Proof of \Cref{lem:width-gives-apx} via \Cref{cl:fractional-packing}]
For any bad triangle $t \in BT$, let $x(t)$ be the expected number of charges of charging scheme $\mc{S}$ to $t$ where the expectation is taken over $\pi$ and the randomization of $\mc{S}$. First, note from the definition of charging schemes that
	\begin{equation}\label{eq:tclrc291}
		\E[|X|] = \sum_{t \in BT} \E\left[x(t)\right].
	\end{equation}	
	Fix a pair $a,b \in V$ of distinct vertices. Since $\mc{S}$ is assumed to have width $w$, we have
	\begin{equation}\label{eq:hcllrc0-9123}
		\sum_{t \in BT: a, b \in t} \E[x(t)] \leq w.
	\end{equation}
	Now let us define $y(t) := \E[x(t)]/w$ for any $t \in BT$. We have
	$$
		\sum_{t \in BT: a, b \in t} y(t) = \frac{1}{w} \sum_{t \in BT: a, b \in t} \E[x(t)] \stackeq{(\ref{eq:hcllrc0-9123})}{\leq} 1.
	$$
	Therefore, we can apply \Cref{cl:fractional-packing} to get 
	\begin{equation}\label{eq:hclrcg}
	\opt{G} \geq \sum_{t \in BT} y(t) \stackeq{(by definition of $y$)}{=} \frac{1}{w} \sum_{t \in BT} \E[x(t)].
	\end{equation} 
	Equations (\ref{eq:tclrc291}) and (\ref{eq:hclrcg}) together prove $\E[|X|] \leq w \cdot \opt{G}$.
\end{myproof}

\end{document}